\documentclass[11pt]{article}
\usepackage{framed}
\usepackage{amssymb,amsfonts,amsmath}
\usepackage{algorithm2e,framed}
\usepackage{amsfonts,amsmath}
\usepackage{graphicx}
\usepackage{subfigure,enumerate,bm,amsmath,amsthm,wrapfig,array,color,algorithmic}
\usepackage{graphics}
\usepackage{url}
\usepackage{multirow}
\usepackage{setspace}

\long\def\symbolfootnote[#1]#2{\begingroup%
\def\thefootnote{\fnsymbol{footnote}}\footnote[#1]{#2}\endgroup}

\newcommand{\Probab}[1]{\mbox{}{\bf{Pr}}\left[#1\right]}

\newcommand{\FNorm }[1]{\mbox{}\|#1\|_{\mathrm{F}}  }
\newcommand{\FNormS}[1]{\mbox{}\|#1\|_{\mathrm{F}}^2}

\newcommand{\TNorm }[1]{\mbox{}\|#1\|_2  }
\newcommand{\TNormS}[1]{\mbox{}\|#1\|_2^2}

\newcommand{\XNorm }[1]{\mbox{}\|#1\|_{\xi}  }
\newcommand{\XNormS}[1]{\mbox{}\|#1\|_{\xi}^2}

\newtheorem{theorem}{\bf Theorem}[]
\newtheorem{lemma}[theorem]{Lemma}

\newcommand{\transp}{^{\textsc{T}}}

\newcommand{\mat}[1]{ {\ensuremath{\mathbf{#1} }}}

\def\rank{\hbox{\rm rank}}

\def\b{{\mathbf b}}

\def\e{{\mathbf e}}
\def\q{{\mathbf q}}

\def\s{{\mathbf s}}
\def\u{{\mathbf u}}
\def\v{{\mathbf v}}

\def\matA{\mat{A}}
\def\matB{\mat{B}}
\def\matC{\mat{C}}
\def\matD{\mat{D}}
\def\matE{\mat{E}}
\def\matH{\mat{H}}
\def\matI{\mat{I}}
\def\matP{\mat{P}}
\def\matQ{\mat{Q}}

\def\matS{\mat{S}}
\def\matT{\mat{T}}
\def\matU{\mat{U}}
\def\matV{\mat{V}}
\def\matW{\mat{W}}
\def\matX{\mat{X}}
\def\matY{\mat{Y}}
\def\matZ{\mat{Z}}
\def\matSig{\mat{\Sigma}}
\def\matOmega{\mat{\Omega}}

\def\scl{{\textsc{l}}}
\def\scu{{\textsc{u}}}

\DeclareMathSymbol{\Prob}{\mathbin}{AMSb}{"50}
\newcommand\remove[1]{}

\def\math#1{$#1$}

\def\mand#1{$$#1$$}
\def\frac#1#2{{#1\over #2}}

\def\mld#1{\begin{equation}
#1
\end{equation}}

\def\eqan#1{\begin{eqnarray*}
#1
\end{eqnarray*}}

\DeclareMathSymbol{\R}{\mathbin}{AMSb}{"52}

\def\choose#1#2{\left({{#1}\atop{#2}}\right)}
\def\cl#1{{\cal #1}}

\def\argmin{\mathop{\hbox{argmin}}\limits}

\def\x{{\mathbf x}}
\def\y{{\mathbf y}}
\def\z{{\mathbf z}}

\def\b{{\mathbf b}}

\def\norm#1{{\|#1\|}}

\def\r#1{{(\ref{#1})}}

\def\matPsi{\mat{\Psi}}

\def\dotfil{\leaders\hbox to 1.5mm{.}\hfill}
\newcounter{rmnum}
\def\RN#1{\setcounter{rmnum}{#1}\uppercase\expandafter{\romannumeral\value{rmnum}}}
\def\rn#1{\setcounter{rmnum}{#1}\expandafter{\romannumeral\value{rmnum}}}

\newcommand{\ignore}[1]{}

\begin{document}

\title{\bf Near-optimal Coresets For Least-Squares Regression}

\author{
Christos Boutsidis\thanks{
Mathematical Sciences Department, IBM T.J. Watson Research Center. Email: cboutsi@us.ibm.com.}
\and
Petros Drineas\thanks{
Computer Science Department,
Rensselaer Polytechnic Institute. Email:
drinep@cs.rpi.edu.}
\and
Malik Magdon-Ismail\thanks{
Computer Science Department,
Rensselaer Polytechnic Institute. Email:
magdon@cs.rpi.edu.}
}

\maketitle
\date{}

\begin{abstract}
We study (constrained) least-squares regression as well as multiple response least-squares regression and ask the question of whether a subset of the data, 
a \emph{coreset},
suffices to compute a good approximate solution to the 
regression. 
We give deterministic, low order polynomial-time algorithms to construct such coresets with approximation guarantees, together with lower bounds indicating that there is not much room for improvement upon our results.
\end{abstract}

\section{Introduction}

Linear regression is an important technique in data analysis~\cite{SLS77}. Research in the area ranges from numerical techniques~\cite{Bjo96} to robustness of the prediction error to noise (e.g., using feature selection~\cite{GE03}).
We ask whether it is
possible to efficiently identify a small subset of the data that contains 
all the essential information of a learning problem. 
Such a subset is called a ``coreset''. We show that
the answer is yes, for linear regression. Such a coreset is analogous to the support vectors in support vector machines~\cite{CS00}. Such coresets contain the meaningful or important points in the data and can be used to find good approximate solutions to the full problem by solving a (much) smaller problem. When the constraints are complex (e.g., non-convex constraints), solving a much smaller regression problem could be a significant saving~\cite{Gao07}.

We present coreset constructions for constrained regression (both simple and multiple response), as well as lower bounds for the size of coresets that achieve certain accuracy. In addition to potential computational savings, a coreset identifies the important core of a machine learning problem and is of considerable interest in applications with huge data where incremental approaches are necessary (for example chunking) and applications where the data is distributed and bandwith is costly (hence communicating only the essential data is imperative~\cite{KV13}).

Our first contribution is a deterministic, polynomial-time algorithm for constructing  a coreset for arbitrarily constrained linear regression. Let \math{k} be the ``effective dimension''  of the data (the rank of the data matrix) and let $\epsilon>0$ be the desired accuracy parameter. Our algorithm constructs a coreset of size \math{O\left(k/\epsilon^2\right)}, which achieves a \math{\left(1+\epsilon\right)}-relative error performance guarantee. In other words, solving the regression problem on the coreset results in a solution which fits all the data with an error which is at most \math{\left(1+\epsilon\right)} worse than the best possible fit to all the data. We extend our results to the setting of multiple response regression using more sophisticated techniques. Our proofs are based
on two sparsification tools from linear algebra \cite{BSS09,BDM11a}, which may be of general interest to the machine learning community, and we discuss these
in some detail.

\subsection{Problem Setup}

Assume the usual setting with \math{n} data points \math{(\z_1,y_1),\ldots,(\z_n,y_n)}; \math{\z_i \in \R^d} are feature vectors (which could have been obtained by applying a non-linear feature transform to raw data) and \math{y_i \in \R} are targets (responses). The linear regression problem asks to determine a vector \math{\x_{opt} \in\cl D\subseteq\R^d} that minimizes
\mand{
\cl E(\x)=\sum_{i=1}^n w_i (\z_i\transp \x-y_i)^2,
}
over \math{\x\in\cl D}, where \math{w_i \in \R} are positive weights. So, \math{\cl E(\x_{opt})\le \cl E(\x)}, for all \math{\x\in\cl D}.
The domain \math{\cl D} represents the constraints on the solution, e.g., in non-negative least squares (NNLS)~\cite{LH74,BMM06}, \math{\cl D=\R^{d}_+}, the nonnegative orthant. Our results
hold for arbitrary~\math{\cl D}.

A coreset of size \math{r < n} is a subset of the data points,
\math{(\z_{i_1},y_{i_1}),\ldots,(\z_{i_r},y_{i_r})}.
The coreset regression problem considers the squared error on the coreset
with a  (possibly) different set of weights \math{s_j > 0},
\mand{
\tilde{\cl E}(\x)=\sum_{j=1}^r s_j (\z_{i_j}\transp\x-y_{i_j})^2.
}
Suppose that \math{\tilde{\cl E}} is minimized at
\math{\tilde\x_{opt} \in\cl D\subseteq\R^d}, so
\math{\tilde{\cl E}(\tilde\x_{opt})\le \tilde{\cl E}(\x)}, for all
\math{\x\in\cl D}. Such a coreset  is of interest if, for some
set of weights \math{s_j}, \math{\tilde\x_{opt}} is nearly as good as \math{\x_{opt}} for the
original regression problem on all the data.
That is, for some small \math{\epsilon>0},
$$
\cl E(\x_{opt}) \leq \cl E(\tilde\x_{opt})\le (1+\epsilon) \cl E(\x_{opt}).
$$
The algorithm which constructs the coreset should also provide the
weights \math{\s_j}.
For the remainder of the paper, we switch to an equivalent
matrix formulation of the problem. (See Appendix for linear algebra background.)

\subsubsection{Matrix Formulation} Let \math{\matA\in\R^{n\times d}} be the data matrix whose rows are the weighted data points \math{\sqrt{w_i}\z_i\transp}; and let \math{\b\in\R^n} be the similarly weighted target vector, \math{b_i=\sqrt{w_i}y_i}, where for $i=1,...,n,$ $b_i$ denotes the $i$th element of $\b \in \R^n$.
The effective dimension of the data can be measured by the rank of \math{\matA}; let \math{k = \rank(\matA)}. Our results hold for arbitrary \math{n>d}, however, in most applications, $n \gg d$ and $\rank(\matA) \approx d$. We can rewrite the squared error as \math{\cl E(\x)=\TNormS{\matA\x-\b}}, so,
\begin{equation}\label{eqn:pppp1}
\x_{opt} \in \argmin_{\x\in\cl D}\TNormS{\matA\x-\b}.
\end{equation}
A coreset of size \math{r < n} is a subset \math{\matC \in \R^{r \times d}} of the rows of \math{\matA} and the corresponding elements \math{\b_c \in \R^r} of \math{\b}. Let \math{\matD\in\R^{r\times r}} be a positive diagonal matrix  for the coreset regression (the weights $s_j$ of the coreset
regression will depend on \math{\matD}).
The weighted squared error on the coreset is given by 
$$\tilde{\cl E}(\x)=\TNormS{\matD(\matC\x-\b_c)},$$ so the coreset regression seeks \math{\tilde\x_{opt}} defined by
$$
\tilde\x_{opt} \in \argmin_{\x\in\cl D}\norm{\matD\left(\matC\x-\b_c\right)}_2^2.
$$
We say that such a coreset is an  \math{(1+\epsilon)}-coreset if the solution obtained by fitting the coreset data is almost optimal for all the data. 
Formally,
\mand{
\TNormS{\matA\x_{opt}-\b}\leq\TNormS{\matA\tilde\x_{opt}-\b}\le
(1+\epsilon)\TNormS{\matA\x_{opt}-\b}.
}

\subsection{Our contributions}

In this section, we discuss our main results for various formulations of
linear regression (also summarized in Table~\ref{table:t1}).  
In the next section we present the relevant algorithms and proofs.

\begin{table}[t]
\centering

\begin{tabular}{|l| c  l|}
  \hline
&&\\[-7pt]
Type of Regression& \multicolumn{2}{c|}{Approximation Ratio} \\[5pt]
\hline\hline
&&\\[-6pt]
  Constrained single-response 
& $1+O(\sqrt{k/r})$ &[Eqn.~(\ref{eqn:pppp1}), Thm.~\ref{lem:regression}] \\[7pt]
Multi-objective & $1+O(\sqrt{k/r})$& [Eqn.~(\ref{eqn:pd1}),
 Thm.~\ref{theorem:multi}]  \\[7pt]
Constrained multiple-response (Frobenius) & $1+O(\sqrt{k\omega/r})$ &[Eqns.~(\ref{eqn:pppp2}), \r{eq:arb-multiple}]  \\[7pt]
Unconstrained multiple-response (Spectral)  & $2+O(\sqrt{\omega/r}+\omega/r+\sqrt{k/r}) $ & [Eqn.~(\ref{eqn:pppp3}), Thm. \ref{lem:regressionG2a}]\\[7pt]
  Unconstrained multiple-response (Frobenius) & $2+O(\sqrt{k/r}) $ 
&[Eqn.~(\ref{eqn:pppp3}), Thm.
\ref{lem:regressionGF}] \\[7pt]
Unconstrained multiple-response ($\b$-agnostic) & 
$O(n/r)$ & [Eqn. (\ref{eqn:pppp3}), Thm.
\ref{lem:regressionG2agnostic}] 

\\[4pt]
    \hline
\end{tabular}

\caption{Summary of our results for
coreset construction in 
linear regression. In all cases, our algorithms are 
\emph{deterministic} and construct a coreset of size \math{r}. 
The approximation ratios are values $\beta$ such that \math{\norm{\matA\tilde\matX_{opt}-\matB}/\norm{\matA\matX_{opt}-\matB} \le \beta}.
In the first row in the table, $\matX_{opt}, \tilde\matX_{opt},$ and $\matB$ are vectors.
{\sc Notation:}  $n$ is the number of data points of dimension $d < n$; 
$k$ is
the rank of the matrix whose rows correspond to the $n$ data points;
\math{r} is the size of the coreset, \math{k<r<n};
 $\omega \geq 1$ is the number of ``response'' vectors
in multiple-response regression (in the last four rows in the table $\matX_{opt}, \tilde\matX_{opt},$ and $\matB$ have $\omega$ columns).
}\label{table:t1}

\end{table}

\subsubsection{Constrained Linear Regression (Section~\ref{sec:simple})}
Our main result for constrained simple regression is Theorem~\ref{lem:regression}, which describes a \emph{deterministic}
 polynomial time algorithm that constructs a \math{(1+\epsilon)}-coreset of size \math{O\left(k/\epsilon^2\right)}. Prior to our work, the best result achieving comparable relative error performance guarantees is Theorem~1 of~\cite{BD09} for constrained regression, and the work of~\cite{DMM06a} for unconstrained regression. Both of these prior
results construct coresets of size \math{O\left(k\log k/\epsilon^2\right)} and they are randomized, so, with some probability, the fit on all the data
can be arbitrarily bad (despite
the coreset being a logarithmic factor larger).
Our methods have comparable, low order polynomial running times and provide deterministic
guarantees. The results in  \cite{DMM06a} and \cite{BD09} were achieved using the matrix concentration results
in~\cite{RV07}. However, these concentration
bounds break unless the coreset size is
 \math{\Omega\left(k \log k/\epsilon^2\right)}.

We extend our results to multiple response regression, where the target is a matrix \math{\matB\in\R^{n\times\omega}} with $\omega \ge 1$. Each column of \math{\matB} is a seperate target (or response) that we wish to predict. We seek to minimize \math{\norm{\matA\matX-\matB}} over all \math{\matX\in\cl D\subseteq\R^{d\times\omega}}. Multiple response regression has numerous applications, but is perhaps most common in multivariate time series analysis; see for example~\cite{Ham94,BF97}. To illustrate, consider prediction of time series data: let \math{\matZ\in\R^{(n+1)\times d}} be a set of \math{d} time series, where each column is a time series with \math{n+1} time steps; we wish to predict time step \math{t+1} from time step~\math{t}. Let \math{\matA} contain the first \math{n} rows of \math{\matZ} and let \math{\matB} contain the last \math{n} rows. Then, we seek \math{\matX} that minimizes \math{\norm{\matA\matX-\matB}_{\xi}} under some
norm \math{\xi}, which is exactly the multiple response regression problem. 
In our work, we consider the
spectral ($\xi=2$) and Frobenius ($\xi=\mathrm{F}$) norms.

\subsubsection{Multi-Objective Regression (Section~\ref{subsec1})} An important variant of multiple response regression is the so-called multi-objective regression. Let 
$$\matB=[\b_1,\ldots,\b_\omega]\in\R^{n\times\omega},$$ where we explicitly identify each column in \math{\matB} as a target response \math{\b_j \in \R^n} where \math{j\in\{1,2,\ldots,\omega\}}. We seek to simultaneously fit multiple target vectors with the same \math{\x}, i.e., to simultaneously minimize \math{\norm{\matA\x-\b_j}_2^2}. This is common when the goal is to trade off different quality criteria simultaneously. Writing \math{\matX=[\x,\x,\ldots,\x]\in\R^{d\times\omega}} (\math{\omega} copies of \math{\x \in\mathcal{D} \subseteq \R^d}), we consider minimizing \math{\norm{\matA\matX-\matB}_{\mathrm{F}}}, which is equivalent to multiple regression with a
strong constraint on \math{\matX}. We present results for coreset constructions for the Frobenius-norm multi-objective regression problem in Theorem~\ref{theorem:multi}, which describes a deterministic algorithm to construct \math{(1+\epsilon)}-coresets of size $O\left(k/\epsilon^2\right)$, where $k=\rank(\matA)$. Theorem~\ref{theorem:multi} emerges by applying Theorem~\ref{lem:regression} after converting the
Frobenius-norm multi-objective regression problem to a simple response regression problem.

\subsubsection{Arbitrarily-Constrained Multiple-Response Regression (Section~\ref{sec:arb})} Using the same approach, converting the problem to a single response regression, we construct a \math{(1+\epsilon)}-coreset for Frobenius-norm arbitrarily-constrained regression in Section \ref{sec:arb}. The coreset size in this case is \math{O\left(k\omega / \epsilon^2\right)}.

\subsubsection{Unconstrained Multiple-Response Regression (Section~\ref{sec:Gen})} In Section~\ref{sec:Gen}, we consider  coresets for unconstrained multiple-response regression for both the spectral and Frobenius norms. The sizes of the coresets are smaller than the constrained case, and our main results are presented in Theorems~\ref{lem:regressionG2a} and~\ref{lem:regressionGF}. Theorem~\ref{lem:regressionG2a} presents a $(2 + \epsilon)$-coreset of size $O((k+\omega)/\epsilon^2)$ for spectral norm regression, while Theorem~\ref{lem:regressionGF} presents a $(2+\epsilon)$-coreset of size $O(k/\epsilon^2)$ for Frobenius norm regression.

\subsubsection{Lower Bounds (Section~\ref{sec:lower})} Finally, in Section~\ref{sec:lower}, we present lower bounds on coreset sizes. In the single response regression setting, we note that our algorithms need to look at the target vector \math{\b}. We show that this is unavoidable, by arguing that no \math{\b}-agnostic deterministic coreset construction algorithm can construct coresets which are small (Theorem~\ref{lem:agnosticD}). We also present similar results for \math{\b}-agnostic randomized coreset constructions (Theorem~\ref{lem:agnosticR}). 

Then, we present lower bounds on the size of coresets for spectral and Frobenius norm multiple response regression that
apply in the general, non \math{\b}-agnostic, setting
(Theorems~\ref{lem:lowerS} and~\ref{lem:lowerF}).

\section{Constrained Linear Regression} \label{sec:simple}

We define constrained linear regression as follows: given $\matA\in\R^{n\times d}$ of rank  $k$, $\b\in\R^n$, and $\mathcal{D} \subseteq \R^d$, we seek \math{\x_{opt}\in\cl D} for which $ \TNormS{\matA \x_{opt} - \b}\le \TNormS{\matA \x - \b}$, for all $\x\in\cl D$ (the domain $\mathcal{D}$ represents the constraints on \math{\x} and can be arbitrary).  To construct a coreset $\matC \in \R^{r \times d}$ (i.e., $\matC$ consists of $r$ rows of $\matA$) and $\b_c \in \R^r$ (i.e., $\b_c$ consists of $r$ elements of $\b$), we introduce \emph{sampling} and \emph{rescaling} matrices $\matS$ and $\matD$ respectively. More specifically, we define the \emph{row-sampling matrix} \math{\matS\in\R^{r\times n}} whose rows are basis vectors \math{\e_{i_1}\transp,\ldots,\e_{i_r}\transp}. Our coreset $\matC$ is now equal to $\matC = \matS\matA$; clearly, \math{\matC} is a matrix whose rows are the rows of \math{\matA} corresponding to indices \math{i_{1},\ldots,i_r}. Similarly, \math{\b_c=\matS\b} contains the corresponding elements of the target vector. Next, let \math{\matD\in\R^{r\times r}} be a positive diagonal rescaling matrix and define the \math{\matD}-weighted regression problem on the coreset as follows:
\begin{equation}\label{eqn:pp1}
\tilde\x_{opt} \in \argmin_{\x\in\cl D}\norm{\matD\left(\matC\x-\b_c\right)}_2^2 = \argmin_{\x\in\cl D}\norm{\matD\matS\left(\matA\x-\b\right)}_2^2.
\end{equation}
In the above, the operator \math{\matD\matS} first samples and then rescales rows of $\matA$ and $\b$. Theorem \ref{lem:regression} is the main result in this section and presents a deterministic algorithm to select a coreset by constructing $\matD$ and $\matS$.
\begin{algorithm}[h]
\begin{framed}
\textbf{Input:} $\matA\in\R^{n\times d}$ of rank $k$, $\b \in \R^n$, and \math{r>k+1}. \\
\noindent \textbf{Output:}
sampling matrix \math{\matS \in \R^{r \times n}} and rescaling matrix \math{\matD \in \R^{r \times r}}.

\begin{algorithmic}[1]
\STATE Compute the SVD of \math{\matY=[\matA,\b]}. Let \math{\matY=\matU\matSig\matV\transp}, where
\math{\matU\in\R^{n\times\ell}}, \math{\matSig\in\R^{\ell\times\ell}}
and
\math{\matV\in\R^{(d+1)\times\ell}}, with \math{\ell\le k +1} (the rank of $\matY$).
\STATE
{\bf Return} \math{[\matD,\matS]=SimpleSampling(\matU,r)}
(see Lemma~\ref{lem:1set})
\end{algorithmic}
\caption{Deterministic coreset construction for
constrained linear regression.\label{alg:SimpReg}}
\end{framed}
\end{algorithm}
\begin{theorem}
\label{lem:regression}
Given $\matA\in\R^{n\times d}$ of rank $k$, $\b\in\R^n$, and $\mathcal{D} \subseteq \R^d$, Algorithm~\ref{alg:SimpReg} constructs matrices \math{\matS\in\R^{r\times n}} and \math{\matD\in\R^{r\times r}} (for any $r > k+1$) such that
 $\tilde{\x}_{opt}$ of Eqn.~(\ref{eqn:pp1}) satisfies
$$
\frac{\TNormS{\matA \tilde{\x}_{opt} - \b}}{\TNormS{\matA \x_{opt} - \b}}
 \leq
\frac{ r+k+1+2\sqrt{r(k+1)}}{r+k+1-2\sqrt{r(k+1)}}
=
1+4\sqrt{\frac{k}{r}}+o\left({\sqrt{{k}/{r}}}\right).
$$
The running time of the proposed algorithm is $T\left(\matU_{\left[\matA,\b\right]}\right)+O\left(r n k^2\right)$, where $T\left(\matU_{\left[\matA,\b\right]}\right)$ is the time needed to compute the left singular vectors of the matrix \math{\left[\matA,\b\right] \in \R^{n \times (d+1)}}.
\end{theorem}
For any \math{0<\epsilon<1}, we can set \math{r=k/\epsilon^2} to get an approximation ratio roughly equal to \math{1+4\epsilon}. This result considerably improves the result in~\cite{BD09}, which needs \math{r=O(k\log k/\epsilon^2)}
to achieve the same approximation ratio. Additionally, our bound is deterministic, whereas the bound in~\cite{BD09} fails with constant probability.
\cite{BD09} also requires an SVD computation in the first step, so its running time is comparable to ours.

In order to prove the above theorem, we need a linear algebraic sparsification result from~\cite{BSS09}, specifically Theorem 3.1 in~\cite{BSS09},
which we restate using our notation (we present the corresponding algorithm below).
\begin{lemma} [Single-set Spectral Sparsification~\cite{BSS09}]
\label{lem:1set}
Given $\matU \in\R^{n\times \ell}$ satisfying \math{\matU\transp\matU=\matI_{\ell}} and \math{r>\ell}, we can deterministically construct sampling and rescaling matrices \math{\matS \in \R^{r \times n}} and \math{\matD \in \R^{r \times r}} such that, for all $\y\in\R^\ell$:
$$ \left(1-\sqrt{{\ell}/{r}}\right)^2 \TNormS{\matU \y}
\leq \TNormS{\matD\matS \matU \y}
\leq \left(1+\sqrt{{\ell}/{r}}\right)^2
\TNormS{\matU \y}.$$
The algorithm runs in $O(r n \ell^2)$ time and we denote it as $[\matD, \matS] = SimpleSampling(\matU, r)$.
\end{lemma}
\begin{proof} (of Theorem \ref{lem:regression})
Let $\matY = \left[\matA,\b\right]\in\R^{n\times(d+1)}$ and compute its SVD: \math{\matY=\matU\matSig\matV\transp}. Let $\ell$ be the rank of $\matY$ ($\ell \leq k +1$, since $\rank(\matA)=k$) and note that \math{\matU\in\R^{n\times\ell}}, \math{\matSig\in\R^{\ell\times\ell}}, and \math{\matV\in\R^{(d+1)\times\ell}}. Let $[\matD, \matS] = SimpleSampling(\matU, r)$ and define \math{\y_1, \y_2\in\R^{\ell}} as follows:
\mand{
\y_1=\matSig\matV\transp
\left[
\begin{matrix}
\x_{opt}\\
-1
\end{matrix}
\right],
\qquad\hbox{and}\qquad
\y_2 = \matSig\matV\transp
\left[
\begin{matrix}
\tilde\x_{opt}\\
 -1
\end{matrix}
\right].
}
Note that \math{\matU\y_1=\matA\x_{opt}-\b}, \math{\matU\y_2=\matA\tilde\x_{opt}-\b}, \math{\matD\matS\matU\y_1=\matD\matS\left(\matA\x_{opt}-\b\right)}, and
\math{\matD\matS\matU\y_2=\matD\matS\left(\matA\tilde\x_{opt}-\b\right)}. We will bound \math{\TNorm{\matU\y_2}} in terms of \math{\TNorm{\matU\y_1}}:
\mand{
 \left(1-\sqrt{{\ell}/{r}}\right)^2
\TNormS{\matU \y_2}
\mathop{\buildrel{(a)}\over{\leq}}
\TNormS{\matD\matS\matU\y_2}
\mathop{\buildrel{(b)}\over{\leq}}
\TNormS{\matD\matS\matU\y_1}
\mathop{\buildrel{(c)}\over{\leq}}
 \left(1+\sqrt{{\ell}/{r}}\right)^2
\TNormS{\matU \y_1}.
}
(a) and (c) follow from Lemma \ref{lem:1set}; (b) follows from the optimality of \math{\tilde\x_{opt}} for the coreset regression in Eqn.~(\ref{eqn:pp1}). Using \math{\ell\le k+1} and manipulating the above expression concludes the proof of the theorem. The running time of the algorithm is equal to the time needed to compute \math{\matU} and the time needed to run the algorithm of Lemma~\ref{lem:1set} with \math{\ell\le k+1}.
\end{proof}
\begin{algorithm}[t] \label{alg:sample2}
\begin{framed}
   \caption{{\sf SimpleSampling} (Lemma~\ref{lem:1set})}
   {\bfseries Input:} $\matU = [\u_1, \u_2, \ldots,\u_n]\transp \in \R^{n \times \ell}$ with $\u_i \in \R^\ell$ and \math{r > \ell}. \\
    {\bfseries Output:} Sampling matrix $\matS \in \R^{r \times n}$ and rescaling matrix $\matD \in \R^{r \times r}$.
 \begin{algorithmic}[1]
    \STATE Initialize  $\matA_0 = \bm{0}_{\ell \times \ell}$, $\matS = \bm{0}_{r \times n}$, and $\matD=\bm{0}_{r \times r}$.
    \STATE Set constants $\delta_L=1$ and $\delta_U=\left( 1 + \ell/r \right) \left(1-\sqrt{\ell/r}\right)^{-1}$.
   \FOR{$\tau=0$ {\bfseries to} $r-1$}
   \STATE Let \math{\scl_\tau=\tau-\sqrt{r \ell}}; 
   \math{\scu_{\tau} = \delta_{U} \left( \tau + \sqrt{\ell r} \right) }.
   \STATE Pick index $i_{\tau}\in\{1,2,...,n\}$ and number $t_{\tau}>0$ (see Section~\ref{sec:ssss} for the definition of $U, L$):
    $$\hspace*{-0.2in} U(\u_{i_{\tau}},\delta_U,\matA_{\tau},\scu_\tau) \le\frac{1}{t_{\tau}}\le
       L(\u_{i_{\tau}}, \delta_L, \matA_{\tau},\scl_\tau).$$
   \STATE Update $\matA_{\tau+1} = \matA_{\tau} + t_{\tau} \u_{i_{\tau}} \u_{i_{\tau}}\transp$;
               and
   set $\matS_{\tau+1, {i_{\tau}}} = 1$, $\matD_{\tau+1,\tau+1} = 1/\sqrt{t_{\tau}}$.
   \ENDFOR
   \STATE Multiply all the weights in $\matD$ by $ \sqrt{ r^{-1} \left(1-\sqrt{\ell/r}\right) }.$
   \STATE {\bfseries Return:} $\matS$ and $\matD$.
\end{algorithmic}
\end{framed}
\end{algorithm}

\subsection{Single-set Spectral Sparsification Algorithm (Lemma~\ref{lem:1set})}\label{sec:ssss}

We now discuss in more detail the sparsification algorithm of Lemma~\ref{lem:1set}. We present the corresponding algorithm as Algorithm~\ref{alg:sample2}.
Our notation deviates from the original in~\cite{BSS09}; we employ our own presentation of the corresponding algorithm in~\cite{BDM11a}.
Algorithm \ref{alg:sample2} is a greedy technique that selects columns one at a time.
To describe the algorithm in more detail, it is convenient to view the input matrix as a set of \math{n} column vectors,
$$\matU\transp=[\u_1, \u_2, \ldots,\u_n],$$ with $\u_i \in \R^{\ell}$ ($i=1,...,n$). Given  $\ell$ and $r>\ell$, introduce the
iterator $\tau = 0, 1,2,...,r-1,$ and define the
parameter $\scl_\tau=\tau-\sqrt{r \ell}$.
For a square symmetric
matrix \math{\matA\in\R^{\ell \times \ell}} with eigenvalues \math{\lambda_1,\ldots,
\lambda_{\ell}}, vector $\u \in \R^{\ell}$ and scalar $\scl\in\R$, define
$$
\phi(\scl, \matA) =
\sum_{i=1}^{\ell}\frac{1}{\lambda_i-\scl},
$$
and let
$L(\u, \delta_L, \matA, \scl)$  be defined as
$$ L(\u, \delta_L, \matA, \scl)  =
\frac{\u\transp(\matA-\scl'\matI_{\ell})^{-2}\u}
{\phi(\scl', \matA)-\phi(\scl,\matA)} -\u\transp(\matA-\scl'\matI_k)^{-1}\u,$$
where $$\scl'= \scl+\delta_L = \scl+1.$$
Similarly, for a square symmetric matrix \math{\matA\in\R^{\ell \times \ell}} with eigenvalues \math{\lambda_1,\ldots,
\lambda_{\ell}}, $\u \in \R^{\ell}$, $\scu\in\R$, define:
$$
\hat{\phi}(\scu,\matA) =
\sum_{i=1}^{\ell}\frac{1}{\scu - \lambda_i},
$$
and let
$U(\u, \delta_{U}, \matA, \scu)$  be defined as
$$ U(\u, \delta_{U}, \matA, \scu)  =
\frac{\u\transp(\matA-\scu'\matI_{\ell})^{-2}\u}
{\hat\phi(\scu,\matA)-\hat{\phi}(\scu', \matA) }
-
\u\transp(\matA-\scu'\matI_{\ell})^{-1}\u,$$
where
$$ \scu' = \scu + \delta_{U} =
\scu + \left( 1 + \ell/r \right) \left(1-\sqrt{\ell/r}\right)^{-1}.$$
The running time of the algorithm
is dominated by the search for an index $i_{\tau}$
satisfying $$\hspace*{-0.2in} U(\u_{i_{\tau}},\delta_U,\matA_{\tau},\scu_\tau) \le\frac{1}{t_{\tau}}\le
       L(\u_{i_{\tau}}, \delta_L, \matA_{\tau},\scl_\tau)$$ (one can achieve that by exhaustive search). One needs
\math{\phi(\scl,\matA)} and \math{\hat\phi(\scl,\matA)}, and hence the eigenvalues of
\math{\matA}. This takes
\math{O(\ell^3)} time, once per iteration, for a total of
\math{O(r\ell^3)}. Then, for \math{i=1,\ldots,n}, we need to compute the functions
\math{L} and \math{U} for every \math{\u_i}. This takes \math{O(n\ell^2)} per iteration,
for a total of \math{O(rn\ell^2)}. So, the total running time of the algorithm is \math{O(nr\ell^2)}.

\section{Constrained Multiple-Response Regression}\label{section:SimpToGen}

Constrained multiple-response regression in the Frobenius norm can be reduced to simple regression. So, we can apply the results of the previous section
to this setting.

\subsection{Multi-Objective Regression}\label{subsec1}
Let $\matA \in \R^{n \times d}$ and $\matB \in \R^{n \times \omega}$, with $\omega \ge 1$. The objective of multi-objective regression is:
\begin{equation}\label{eqn:pd1}
\min_{\x\in\cl D} \norm{\matA[\x,\ldots,\x]-\matB}_{\mathrm{F}}^2,
\end{equation}
where $[\x,\ldots,\x] \in \R^{d \times \omega}$ contains $\omega$ copies of $\x \in \mathcal{D} \subseteq \R^d$. Let \math{\b_{avg}=\frac{1}{\omega}\matB\bm1_\omega} (here $\bm1_\omega \in \R^{\omega}$ is a vector of all ones and thus $\b_{avg} \in \R^n$ is the average of the columns in \math{\matB}).  Recall that \math{\matA\in\R^{n\times d}}, \math{\matB\in\R^{n\times\omega}}, and let \math{\matX=[\x,\ldots,\x]\in\R^{d\times\omega}}.
\begin{lemma}
\label{lemma:multi2simple1} For
\math{\matX=[\x,\ldots,\x]\in\R^{d\times\omega}},
\math{\displaystyle \norm{\matA\matX-\matB}_\mathrm{F}^2=\omega\norm{\matA\x-\b_{avg}}_2^2
+\sum_{i=1}^\omega\norm{\b_{avg}-\matB^{(i)}}_2^2}.
\end{lemma}
In the above, $\matB^{(i)} \in \R^n$ denotes the $i$-th column of $\matB$ as a column vector. Note that the second term in Lemma \ref{lemma:multi2simple1} does not depend on \math{\x} and thus the generalized multi-objective regression can be reduced to simple regression on \math{\matA} and \math{\b_{avg}}. Using Theorem~\ref{lem:regression}, we can get a coreset: let \math{\tilde\x_{opt}} minimize
\math{\TNorm{\matD\matS\left(\matA\x-\b_{avg}\right)}}, where \math{\matS} and \math{\matD} are obtained via Theorem~\ref{lem:regression} applied to \math{\matA} and $\b_{avg}$.
If \math{\tilde\matX_{opt}=\left[\tilde\x_{opt},\ldots,\tilde\x_{opt}\right]}, then, by
Lemma~\ref{lemma:multi2simple1}, \math{\tilde\matX_{opt}} minimizes \math{\norm{\matD\matS\left(\matA\matX-\matB\right)}_{\mathrm{F}}}. Similarly, if \math{\x_{opt}} minimizes
\math{\TNorm{\matA\x-\b_{avg}}} and \math{\matX_{opt}=\left[\x_{opt},\ldots,\x_{opt}\right]}, then \math{\matX_{opt}} minimizes \math{\norm{\matA\matX-\matB}_{\mathrm{F}}}.
Theorem \ref{theorem:multi} states that \math{\tilde\matX_{opt}} approximates \math{\matX_{opt}}.

\begin{theorem}\label{theorem:multi}
Given $\matA\in\R^{n\times d}$ of rank $k$ and $\matB\in\R^{n \times \omega}$, we can construct matrices \math{\matS\in\R^{r\times n}} and \math{\matD\in\R^{r\times r}} (for any $r > k+1$) such that the matrix $\tilde\matX_{opt}=\left[\tilde\x_{opt},\ldots,\tilde\x_{opt}\right]$ that minimizes \math{\norm{\matD\matS\left(\matA\matX-\matB\right)}_\mathrm{F}} over all matrices $\matX=\left[\x, \x,\ldots,\x\right]$ with $\x \in \mathcal{D} \subseteq \R^d$ satisfies:
\mand{
\norm{\matA \tilde{\matX}_{opt} - \matB}_\mathrm{F}^2
\le
\left(1+O\left(\sqrt{{k}/{r}}\right)\right)\norm{\matA \matX_{opt} -
\matB}_\mathrm{F}^2.}
The run time of the proposed algorithm is $T\left(\matU_{\left[\matA,\b_{avg}\right]}\right)+O\left(n\omega+ r n k^2\right)$, where $T\left(\matU_{\left[\matA,\b_{avg}\right]}\right)$ is the time needed to compute the left singular vectors of the matrix \math{\left[\matA,\b_{avg}\right] \in \R^{n \times (d+1)}}.
\end{theorem}
\begin{proof} 
We first construct \math{\matD} and \math{\matS} via Theorem~\ref{lem:regression} applied to \math{\matA} and \math{\b_{avg}}. The running time is \math{O\left(n\omega\right)}
(the time needed to compute \math{\b_{avg}}) plus the running time of Theorem~\ref{lem:regression}. The result is immediate from the following derivation:
\eqan{
\norm{\matA \tilde{\matX}_{opt} - \matB}_\mathrm{F}^2
&\buildrel(a)\over=&
\omega\norm{\matA\tilde\x_{opt}-\b_{avg}}^2
+\sum_{i=1}^\omega\norm{\b_{avg}-\matB^{(i)}}^2
\\
&\buildrel(b)\over\le&
\left(1+O\left(\sqrt{k/r}\right)\right)^2\omega\norm{\matA\x_{opt}-\b_{avg}}^2
+\sum_{i=1}^\omega\norm{\b_{avg}-\matB^{(i)}}^2
\\
&\le&
\left(1+O\left(\sqrt{k/r}\right)\right)^2\left(\omega\norm{\matA\x_{opt}-\b_{avg}}^2
+\sum_{i=1}^\omega\norm{\b_{avg}-\matB^{(i)}}^2\right)
\\
&\buildrel(a)\over=&
\left(1+O\left(\sqrt{k/r}\right)\right)^2
\norm{\matA\matX_{opt} - \matB}_\mathrm{F}^2.
}
(a) follows by Lemma~\ref{lemma:multi2simple1}; (b) follows because
\math{\tilde\x_{opt}} is the output of a coreset regression as in Theorem~\ref{lem:regression}. Finally, \math{r>k+1} implies that \math{\left(1+O\left(\sqrt{k/r}\right)\right)^2=1+O\left(\sqrt{k/r}\right)}.
\end{proof}

\subsection{Arbitrarily-Constrained Multiple-Response Regression}\label{sec:arb}

Multi-objective regression is a special case of constrained multiple-response regression for which we can efficiently obtain the coresets. In the general case,  the problem still reduces to
simple regression, but the coresets are now larger. The objective of arbitrarily-constrained multiple-response regression is
\begin{equation}\label{eqn:pppp2}
\min_{\matX\in\cl D\subseteq\R^{d\times\omega}} \norm{\matA\matX-\matB}_\mathrm{F}. 
\end{equation}
Since \math{\R^{d\times\omega}}  is isomorphic to
\math{\R^{d\omega}}, we can view \math{\matX\in\R^{d\times\omega}} as a ``stretched out'' vector \math{\hat\matX\in\R^{d\omega}}; corresponding to the domain \math{\cl D} is the domain \math{\hat{\cl D}\subseteq\R^{d\omega}}. Similarly, we can stretch out  \math{\matB\in\R^{n\times \omega}} to \math{\hat\matB\in\R^{n\omega}}. To complete the transformation to simple linear regression, we build a transformed block-diagonal data matrix \math{\hat\matA} from \math{\matA}, by repeating \math{\omega} copies of \math{\matA} along the diagonal:
\mand{
\hat\matA
=
\left[
\begin{matrix}
\matA&&&\\
&\matA&&\\
&&\ddots\\
&&&\matA
\end{matrix}
\right]
\in\R^{n\omega\times d\omega},
\qquad
\hat\matX=
\left[
\begin{matrix}
\matX^{(1)}\\
\matX^{(2)}\\
\vdots\\
\matX^{(\omega)}\\
\end{matrix}
\right]
\in\R^{d\omega},
\qquad
\hat\matB=
\left[
\begin{matrix}
\matB^{(1)}\\
\matB^{(2)}\\
\vdots\\
\matB^{(\omega)}\\
\end{matrix}
\right]
\in\R^{n\omega}.
}
\begin{lemma}
\label{lemma:Gen2Simp}
For all \math{\matA}, \math{\matX} and \math{\matB} of appropriate dimensions,
\math{\FNormS{\matA\matX-\matB}=\TNorm{\hat\matA\hat\matX-\hat\matB}^2}.
\end{lemma}

\noindent Theorem~\ref{lem:regression} gives us
 coresets for this equivalent regression. Note that  \math{\rank(\hat\matA)\le \omega\cdot\rank(\matA)}. The coreset will identify
the important rows of \math{\matA} (the same row may get identified multiple times as different rows of \math{\hat\matA}), and the important \emph{elements} of \math{\matB}, because the entries in \math{\hat\matB} are elements of \math{\matB}, not rows of \math{\matB}. Let \math{\hat\matX_{opt}} be the solution constructed from the coreset, which minimizes \math{\norm{\hat\matA\hat\matX-\hat\matB}} over \math{\hat\matX\in\hat{\cl D}}, and let \math{\tilde\matX_{opt}\in{\cl D}} be the corresponding solution in the original domain \math{\cl D}. If \math{r} is the size of the coreset and \math{\rank(\matA)=k}, then, by Theorem~\ref{lem:regression},
\mld{
\norm{\matA\tilde\matX_{opt}-\matB}_\mathrm{F}^2
\le
\left(1+O\left(\sqrt{k\omega/r}\right)\right)
\FNormS{\matA\matX_{opt}-\matB}.
\label{eq:arb-multiple}
}
So, for the approximation ratio to be \math{1 + O(\epsilon)}, we set \math{r = O\left(k\omega/\epsilon^2\right)}. The running time would involve the time needed to compute the SVD of \math{[\hat\matA,\hat\matB]}.

Notice that the coresets are large and somewhat costly to compute and
they only work for the Frobenius norm. In the next section, using more sophisticated techniques, we will get smaller coresets for unconstrained regression in both the Frobenius and spectral norms.


\begin{algorithm}[h]
\begin{framed}
\textbf{Input:} $\matA\in\R^{n\times d}$ of rank $k$, \math{\matB\in\R^{n\times\omega} },
and  \math{r>k}. \\
\noindent \textbf{Output:}
sampling matrix \math{\matS \in \R^{r \times n}} and rescaling matrix \math{\matD \in \R^{r \times r}}.

\begin{algorithmic}[1]
\STATE Compute the SVD of \math{\matA}: \math{\matA=\matU_{\matA}\matSig_{\matA}\matV_{\matA}\transp}, where
\math{\matU_{\matA}\in\R^{n\times k}}, \math{\matSig_{\matA}\in\R^{k\times k}}, and \math{\matV_{\matA}\in\R^{d\times k}}; compute
\math{\matE=\matU_{\matA}\matU_{\matA}\transp\matB-\matB}.
\STATE
{\bf return} \math{[\matS,\matD]=MultipleSpectralSampling(\matU_{\matA},\matE,r)}
(see Lemma~\ref{lem:2setS})
\end{algorithmic}
\caption{Deterministic coresets for
multiple regression in spectral norm.}
\label{alg:GenS}
\end{framed}
\end{algorithm}

\section{Unconstrained Multiple-Response Regression}
\label{sec:Gen}

Consider the following problem: given a matrix $\matA\in\R^{n \times d}$ with rank $k$ and a matrix $\matB\in\R^{n \times \omega}$ with $\omega \ge 1$, we seek to identify the matrix $\matX_{opt}\in\R^{d\times\omega}$ that satisfies ($\xi =2$ and $\xi=\mathrm{F}$)
\begin{equation}\label{eqn:pppp3}
\matX_{opt} \in \arg\min_{\matX\in\R^{d\times\omega}}\norm{\matA\matX-\matB}_{\xi}^2.
\end{equation}
We can compute \math{\matX_{opt}} via the pseudoinverse of \math{\matA}, namely $\matX_{opt} = \matA^{\dagger}\matB$. If $\matS$ and $\matD$ are sampling and rescaling matrices respectively, then the coreset regression problem is:
\begin{equation}\label{eqn:pp2}
\tilde{\matX}_{opt} \in \arg\min_{\matX\in\R^{d\times\omega}}\norm{\matD\matS \left(\matA \matX - \matB\right)}_{\xi}^2 = \arg\min_{\matX\in\R^{d\times\omega}}\norm{\matD\matS \matA \matX - \matD\matS\matB}_{\xi}^2.
\end{equation}
The solution of the coreset regression problem is  $\tilde{\matX}_{opt} = \left(\matD\matS\matA\right)^{\dagger}\matD\matS\matB$. The main results in this section are presented in
Theorems~\ref{lem:regressionG2a} and~\ref{lem:regressionGF}.

\begin{theorem} [Spectral norm]\label{lem:regressionG2a}
Given a matrix $\matA\in\R^{n \times d}$ with rank $k$, a matrix $\matB\in\R^{n \times \omega}$, and $r > k$, Algorithm~\ref{alg:GenS} deterministically constructs  matrices
$\matS \in \R^{r \times n}$ and $\matD \in \R^{r \times r}$ such that the solution of the problem of Eqn.~(\ref{eqn:pp2}) satisfies: 
$$\TNormS{\matA \tilde{\matX}_{opt} - \matB} \leq \TNormS{\matA \matX_{opt} - \matB} +
\left( \frac{ 1 + \sqrt{{\omega}/{r}} }{1 - \sqrt{{k}/{r}}} \right)^2 \TNormS{\matA \matX_{opt} - \matB}.$$
The running time of the proposed algorithm is $T\left(\matU_{\matA}\right)+O\left(rn\left(k^2+\omega^2\right)\right)$, where $T\left(\matU_{\matA}\right)$ is the time needed to compute the left singular vectors of $\matA$.
\end{theorem}
Since \math{r>k}, the approximation ratio is 
\math{2+O(\sqrt{{\omega}/{r}}+{\omega}/{r}+\sqrt{{k}/{r}})}.
So, for $\epsilon>0$ and $r = O((\omega+k)/\epsilon^2)$ the approximation ratio is
$2+\epsilon$.
For \math{r>\omega}, the approximation is \math{O(1)}, while for 
\math{r<\omega},  is asymptotic to 
\math{O\left(\omega/r\right)}. 
We will argue that this is nearly optimal by providing a 
matching lower bound in Theorem~\ref{lem:lowerS}. 

\begin{theorem} [Frobenius norm]\label{lem:regressionGF}
Given matrix $\matA\in\R^{n \times d}$ of rank  $k$, matrix
$\matB\in\R^{n \times \omega}$, and $r > k,$ Algorithm~\ref{alg:GenF} deterministically constructs a sampling matrix $\matS \in \R^{r \times n}$ and a rescaling matrix $\matD \in \R^{r \times r}$ such that the solution of the problem of Eqn.~(\ref{eqn:pp2}) satisfies: 
$$\FNormS{\matA \tilde{\matX}_{opt} - \matB} \leq \FNormS{\matA \matX_{opt} - \matB} +
\frac{1}{\left(1 - \sqrt{{k}/{r}} \right)^2} \FNormS{\matA \matX_{opt} - \matB}.$$
The running time of the proposed algorithm is $T\left(\matU_{\matA}\right)+O\left(rnk^2\right)$, where $T\left(\matU_{\matA}\right)$ is the time needed to compute the left singular vectors of $\matA$.
\end{theorem}
\begin{algorithm}[h]
\begin{framed}
\textbf{Input:} $\matA\in\R^{n\times d}$ of rank $k$, \math{\matB\in\R^{n\times\omega} },
and  \math{r>k}. \\
\noindent \textbf{Output:}
sampling matrix \math{\matS \in \R^{r \times n}} and rescaling matrix \math{\matD \in \R^{r \times r}}.

\begin{algorithmic}[1]
\STATE Compute the SVD of \math{\matA}: \math{\matA=\matU_{\matA}\matSig_{\matA}\matV_{\matA}\transp}, where
\math{\matU_{\matA}\in\R^{n\times k}}, \math{\matSig_{\matA}\in\R^{k\times k}}, and \math{\matV_{\matA}\in\R^{d\times k}}; compute
\math{\matE=\matU_{\matA}\matU_{\matA}\transp\matB-\matB}.
\STATE
{\bf return} \math{[\matS,\matD]=MultipleFrobeniusSampling(\matU_{\matA},\matE,r)}
(see Lemma~\ref{lem:2setF})
\end{algorithmic}
\caption{Deterministic coresets for
multiple regression in Frobenius norm.}
\label{alg:GenF}
\end{framed}
\end{algorithm}

The approximation ratio in the above theorem is
\math{2+O(\sqrt{k/r})}. In Theorem~\ref{lem:lowerF}, we will give
 a lower bound for the approximation ratio which is \math{1+\Omega(k/r)}. We conjecture that our lower bound can be achieved (deterministically), perhaps by a more 
sophisticated  algorithm or analysis. 

Finally, we note that the \emph{$\matB$-agnostic}
randomized construction of~\cite{DMM08}
achieves a $(1+\epsilon)$ approximation ratio using a significantly
larger coreset,
$r=O(k \log k/\epsilon^2)$.
Importantly,~\cite{DMM08} does not need any access to $\matB$ in order to construct the coreset, whereas our approach constructs coresets by carefully choosing important data points with respect to the particular target response matrix \math{\matB}. We will also discuss $\matB$-agnostic algorithms in Section~\ref{sec:agnostic} (Theorem~\ref{lem:regressionG2agnostic}) and we will present matching lower bounds in Section~\ref{sec:lower}.

\subsection{Proofs of Theorems \ref{lem:regressionG2a} and \ref{lem:regressionGF}}

We will make heavy use of
facts from Section~\ref{sec:secBack} in the Appendix. We start with a few simple lemmas.
\begin{lemma} \label{lem:simple}
Let $\matE=\matA \matX_{opt} - \matB \in \R^{n \times \omega}$ be the regression residual. Then, \math{\rank(\matE)\le\min\{\omega, n-k\}}.
\end{lemma}
\begin{proof}
Using our notation, $\matA \matX_{opt} - \matB=-\left(\matI_n-\matU_\matA \matU_\matA\transp\right)\matB=-\matU_\matA^{\perp} \left(\matU_\matA^{\perp}\right)\transp \matB$. To conclude notice that
\math{\rank(\matX\matY) \le \min\{\rank(\matX),\rank(\matY)\}} for any matrices $\matX$ and $\matY$.
\end{proof}

\noindent We now present our main tool for obtaining approximation guarantees for coreset regression.
\begin{lemma}\label{lem:generic}
Assume that the rank of the matrix $\matD\matS \matU_\matA \in \R^{r \times k}$ is equal to $k$ (i.e., the matrix has full rank). Then, for $\xi = 2,\mathrm{F}$,
$$ \norm{ \matA \tilde{\matX}_{opt} - \matB }_\xi^2 \leq \XNormS{\matA \matX_{opt}-\matB} + \XNormS{(\matD\matS\matU_\matA)^{\dagger}\matD\matS\left(\matA \matX_{opt}-\matB\right)}.$$
\end{lemma}
\begin{proof} 
To simplify notation, let \math{\matW=\matD\matS}. Using the SVD of $\matA$, \math{\matA=\matU_\matA \matSig_\matA \matV_\matA\transp}, we get:
\mand{
\norm{\matB - \matA \tilde{\matX}_{opt}}_\xi^2
=\norm{\matB -  \matU_\matA \matSig_\matA \matV_\matA\transp
(\matW\matU_\matA \matSig_\matA \matV_\matA\transp)^{\dagger}\matW\matB}_\xi^2
=
\norm{\matB -  \matU_\matA
(\matW\matU_\matA)^{\dagger}\matW\matB}_\xi^2,
}
where the last equality follows from properties of the pseudo-inverse and the fact that $\matW\matU_\matA$ is a full-rank matrix (see Lemma~\ref{lem:pseudo} in the Appendix).
Using \math{\matB=\left(\matU_\matA \matU_\matA\transp +
\matU_\matA^{\perp} \left(\matU_\matA^{\perp}\right)\transp \right)\matB}, we obtain
\eqan{
\norm{\matB - \matA \tilde{\matX}_{opt}}_\xi^2 &=& \norm{\matB -  \matU_\matA \left(\matW\matU_\matA\right)^{\dagger}\matW\left(\matU_\matA \matU_\matA\transp +
\matU_\matA^{\perp} \left(\matU_\matA^{\perp}\right)\transp \right)\matB}_\xi^2\\
&=& \norm{\matB - \matU_\matA  \left(\matW \matU_\matA\right)^{\dagger}\matW \matU_\matA \matU_\matA\transp \matB + \matU_\matA \left(\matW \matU_\matA\right)^{\dagger}\matW
\matU_\matA^{\perp} \left(\matU_\matA^{\perp}\right)\transp \matB}_\xi^2 \\
\nonumber  &\buildrel{(a)}\over{=}&
\XNormS{\matU_{\matA}^\perp (\matU_{\matA}^\perp)\transp\matB +
\matU_\matA (\matW\matU_\matA)^{\dagger}\matW \matU_\matA^{\perp}
(\matU_\matA^{\perp})\transp \matB} \\
\nonumber  &\buildrel{(b)}\over{\leq}&
\XNormS{\matU_{\matA}^\perp (\matU_{\matA}^\perp)\transp\matB} +
\XNormS{\matU_\matA (\matW\matU_\matA)^{\dagger}\matW \matU_\matA^{\perp}
(\matU_\matA^{\perp})\transp \matB}.
}
\math{(a)} follows from the assumption that the rank of \math{\matW\matU_\matA} is equal to $k$ and thus $ \left(\matW \matU_\matA\right)^{\dagger}\matW \matU_\matA = \matI_{k}$ and \math{(b)} follows by matrix-Pythagoras (Lemma \ref{lem:pythagoras}). To conclude, we use spectral submultiplicativity on the second term and the fact that \math{\matU_\matA^{\perp} \left(\matU_\matA^{\perp}\right)\transp \matB= -(\matA\matX_{opt}-\matB)}.
\end{proof}
\noindent This lemma provides a framework for coreset construction:
all we need are sampling and rescaling matrices $\matS$ and  $\matD$, such that $\rank(\matD\matS \matU_\matA)=k$ and 
$$\XNormS{\left(\matD\matS\matU_\matA\right)^{\dagger}\matD\matS\left(\matA \matX_{opt}-\matB\right)}$$ is small. The final ingredients for the proofs of Theorems \ref{lem:regressionG2a} and \ref{lem:regressionGF} are two matrix sparsification results that we present in the Appendix. 
\begin{lemma} 
\label{lem:2setS}
Let \math{\matY\in\R^{n\times \ell_1}} and \math{\matPsi\in\R^{n\times \ell_2}} with respective ranks \math{\rho_{\matY}}, and \math{\rho_\matPsi}. Given \math{r>\rho_\matY}, there exists a deterministic algorithm that runs in time $T_{SVD}\left(\matY\right) + T_{SVD}\left(\matPsi\right) + O(r n (\rho_{\matY}^2+\rho_{\matPsi}^2))$ and constructs sampling and rescaling matrices \math{\matS\in\R^{r\times n}},
\math{\matD\in\R^{r\times r}} satisfying:
\mand{
\rank\left(\matD\matS \matY\right) = \rank\left(\matY\right);
\
\TNorm{\left(\matD\matS\matY\right)^{\dagger}} < \frac{1}{1 - \sqrt{{\rho_{\matY}}/{r} }}
 \TNorm{\matY^{\dagger}};
\
\TNorm{\matD\matS\matPsi}
<\left( 1 + \sqrt{ \frac{\rho_{\matPsi}}{r} }\right) \TNorm{\matPsi}.
}
If $\matPsi = \matI_{n}$, the running time of the algorithm reduces to $T_{SVD}\left(\matY\right)+O\left(r n \rho_{\matY}^2\right)$. We write $\left[\matD, \matS\right] = MultipleSpectralSampling\left(\matY,\matPsi, r\right)$ to denote such a deterministic procedure.
\end{lemma}
\begin{lemma} 
\label{lem:2setF}
Let \math{\matY\in\R^{n\times \ell_1}} and \math{\matPsi\in\R^{n\times \ell_2}} with respective ranks \math{\rho_{\matY}}, and \math{\rho_\matPsi}. Given \math{r>\rho_\matY}, there exists a deterministic algorithm that runs in time $T_{SVD}(\matY)+O(r n \rho_{\matY}^2  + \ell_2 n)$ and constructs sampling and rescaling matrices \math{\matS\in\R^{r\times n}},
\math{\matD\in\R^{r\times r}} satisfying:
\mand{
\rank\left(\matD\matS \matY\right) = \rank\left(\matY\right);
\
\TNorm{\left(\matD\matS\matY\right)^{\dagger}} < \frac{1}{1 - \sqrt{{\rho_{\matY}}/{r}}}
 \TNorm{\matY^{\dagger}};  \quad
\
\FNorm{\matD\matS\matPsi} \le \FNorm{\matPsi}.
}
If $\matPsi = \matI_{n}$, the running time of the algorithm reduces to $T_{SVD}\left(\matY\right)+O\left(r n \rho_{\matY}^2\right)$. We write $\left[\matD, \matS\right] = MultipleFrobeniusSampling\left(\matY, \matPsi, r\right)$ to denote such a deterministic procedure.
\end{lemma}

\begin{proof} (of Theorem~\ref{lem:regressionG2a})
Theorem~\ref{lem:regressionG2a} follows from Lemmas~\ref{lem:generic} and~\ref{lem:2setS}. First, compute the SVD of \math{\matA} to obtain \math{\matU_{\matA} \in \R^{n \times k}}, and
let $\matE = \matA \matX_{opt} - \matB=\matU_{\matA}\matU_{\matA}\transp\matB-\matB$. Next, run the algorithm of Lemma~\ref{lem:2setS} to obtain $\left[\matD, \matS\right] = MultipleSpectralSampling\left(\matU_{\matA}, \matE, r\right)$. This algorithm runs in time $T_{SVD}\left(\matE\right)+ O\left( r n \left(k^2+ \rho_{\matE}^2\right) \right)$, where
\math{k} is the rank of $\matU_\matA$ and $\matA$. The total running time of the algorithm is \math{T(\matU_\matA)+T_{SVD}\left(\matE\right)+O\left( r n \left(k^2+ \rho_{\matE}^2\right) \right)
=T\left(\matU_\matA\right)+O\left(rn\left(k^2+\omega^2\right)\right)}.

Lemma \ref{lem:2setS} guarantees that $\matD$ and $\matS$ satisfy the rank assumption of Lemma \ref{lem:generic}. To conclude the proof, we bound the second term of Lemma~\ref{lem:generic}, using
the bounds of Lemma \ref{lem:2setS} and \math{\rho_\matE\le\min\{\omega, n-k\}\le\omega}:
\eqan{
\TNormS{
(\matD\matS\matU_\matA)^{\dagger}\matD\matS\left(\matA \matX_{opt}-\matB\right)}
&\le&
\TNormS{
\left(\matD\matS\matU_\matA\right)^{\dagger}}\TNormS{\matD\matS\left(\matA \matX_{opt}-\matB\right)}
\\
&\le&
\left(1 - \sqrt{{k}/{r} }\right)^{-2} \left(1 + \sqrt{{\omega}/{r}}\right)^2 \TNormS{\matA \matX_{opt}-\matB} .
}
\end{proof}

\begin{proof} (of Theorem~\ref{lem:regressionGF})
The proof is similar to the proof of Theorem~\ref{lem:regressionG2a}, using Lemma~\ref{lem:2setF} instead of Lemma~\ref{lem:2setS}.
Let $\left[\matD, \matS\right] = MultipleFrobeniusSampling\left(\matU_{\matA}, \matE, r\right)$
We bound the second term of Lemma~\ref{lem:generic}, using
the bounds of Lemma \ref{lem:2setF}:
\eqan{
\FNormS{
(\matD\matS\matU_\matA)^{\dagger}\matD\matS\left(\matA \matX_{opt}-\matB\right)}
&\le&
\TNormS{
\left(\matD\matS\matU_\matA\right)^{\dagger}}\FNormS{\matD\matS\left(\matA \matX_{opt}-\matB\right)}
\\
&\le&
\left(1 - \sqrt{{k}/{r} }\right)^{-2}  \FNormS{\matA \matX_{opt}-\matB} .
}

\end{proof}

\subsection{\math{\matB}-Agnostic Coreset Construction}
\label{sec:agnostic}

All the coreset construction algorithms that we presented so far carefully construct the coreset using knowledge of the response vector. If the algorithm does not need knowledge of \math{\matB} to construct the coreset, and yet can provide an approximation guarantee for every \math{\matB}, then the algorithm is \math{\matB}-agnostic. A \math{\matB}-agnostic coreset construction algorithm is appealing because the coreset, as specified by the sampling and rescaling matrices \math{\matS} and \math{\matD}, can be computed off-line and applied to any \math{\matB}.  We briefly digress to show how our methods can be extended to develop \math{\matB}-agnostic coreset constructions.

\begin{theorem} [\math{\matB}-agnostic Coresets] \label{lem:regressionG2agnostic}
Given a matrix $\matA\in\R^{n \times d}$ with rank $k$, a matrix $\matB\in\R^{n \times \omega}$, and $r > k$, there exists an algorithm to deterministically construct a sampling matrix $\matS$ and a rescaling matrix $\matD$ such that for any $\matB \in \R^{n \times \omega}$, the matrix $\tilde{\matX}_{opt}$ that solves the problem of Eqn.~(\ref{eqn:pp2}) satisfies: 
$$\XNormS{\matA \tilde{\matX}_{opt} - \matB} \leq \XNormS{\matA \matX_{opt} - \matB} +
\left( \frac{ 1 + \sqrt{{n}/{r}} }{1 - \sqrt{{k}/{r}}} \right)^2 \XNormS{\matA \matX_{opt} - \matB}.$$
The running time of the proposed algorithm is $T\left(\matU_{\matA}\right)+O\left(rnk^2\right)$, where $T\left(\matU_{\matA}\right)$ is the time needed to compute the left singular vectors of $\matA$.
\end{theorem}
\begin{proof}
The proof is similar
to the proof of Theorem~\ref{lem:regressionG2a}, except we now construct the sampling and rescaling matrices as
$\left[\matS, \matD \right] = MultipleSpectralSampling\left(\matU_{\matA}, \matI_{n}, r\right)$.
To bound the second term in Lemma~\ref{lem:generic}, we use
\eqan{
\XNormS{
\left(\matD\matS\matU_\matA\right)^{\dagger}\matD\matS\left(\matA \matX_{opt}-\matB\right)}
&=&
\XNormS{
\left(\matD\matS\matU_\matA\right)^{\dagger}\matD\matS\matI_n\left(\matA \matX_{opt}-\matB\right)}
\\
&\le&
\TNormS{
\left(\matD\matS\matU_\matA\right)^{\dagger}}\TNormS{\matD\matS\matI_n}
\XNormS{\left(\matA \matX_{opt}-\matB\right)},
}
and the bounds of Lemma~\ref{lem:2setS}.
\end{proof}

\noindent The above bound decreases with \math{r} and holds for any
 \math{\matB}, guaranteeing a constant-factor approximation
with a constant fraction of the data.
The approximation ratio is \math{O(n/r)}, which  seems quite weak. In the next section, we show that
 this result is indeed tight.

\section{Lower Bounds on Coreset Size}\label{sec:lower}

We have just seen a \math{\matB}-agnostic coreset construction algorithm
with a rather weak worst case guarantee of \math{O(n/r)} approximation error. 
We will now show that no deterministic \math{\matB}-agnostic 
coreset construction algorithm can guarantee 
a better error (Theorem~\ref{lem:agnosticD}) by providing lower bounds
on coreset size as a function of approximation error. These results
are also summarized in Table~\ref{table:t2}.

\remove{
\begin{table}
\begin{center}
\begin{tabular}{|l| c | c| c|}
  \hline
&&&\\[-7pt]
Type of Regression & Desired Approximation &Lower bound & Best Known Coreset Size \\[4pt]
\hline\hline
&&&\\[-8pt]
 Deterministic \math{\b}-agnostic                &\math{(1+\epsilon)} & $\frac{n}{1+\epsilon}$ [Thm. \ref{lem:agnosticD}]& $O(\frac{n}{1+\epsilon})$ [Thm. \ref{lem:regressionG2agnostic}]\\[7pt]
Randomized \math{\b}-agnostic                  &\math{(1+\epsilon)} w.p. \math{\frac12} & $\Omega(\frac{1}{\epsilon})$ [Thm. \ref{lem:agnosticR}] & $O(\frac{k\log k}{\epsilon^2})$ [Thm. 5 in~\cite{DMM08}]\\[7pt]
Multiple-regression ($\xi=\mathrm{F}$)      &\math{O(1)}& $\Omega(d)$ [Thm. \ref{lem:lowerF}] & $O(d)$ [Thm. \ref{lem:regressionGF}]\\[7pt]
Multiple-regression ($\xi=2$) &\math{O(1)}& $\Omega(\omega)$ [Thm. \ref{lem:lowerS}] &$O(d+\omega)$ [Thm. \ref{lem:regressionG2a}] \\[7pt]
\hline
\end{tabular}
\end{center}
\caption{Lower bounds on \emph{coreset size} for different formulations of
linear regression. {\sc Notation:} \math{d\ge1}  is the dimension and 
rank of 
the data matrix \math{\matA \in \R^{n \times d}};  \math{\omega} is the number of response 
vectors in the multiple-regression; \math{r>d+\omega} is the coreset size;
\math{\epsilon>0} is an accuracy parameter.
\label{table:t2}}
\end{table}
}
\begin{table}
\begin{center}
\begin{tabular}{|l| cl| cl|}
  \hline
&&&&\\[-7pt]
Type of Regression & \multicolumn{2}{c|}{Lower bound} & \multicolumn{2}{c|}{Known Approximation Ratio} \\[6pt]
\hline\hline
&&&&\\[-6pt]
 Deterministic \math{\b}-agnostic & $n/r$ &[Thm. \ref{lem:agnosticD}]& $O({n}/{r})$ &[Thm. \ref{lem:regressionG2agnostic}]\\[7pt]
Randomized \math{\b}-agnostic & 1+$\Omega({1}/{r})$&
[Thm. \ref{lem:agnosticR}] & $1+O(\sqrt{{k\log k}/{r}})$ &[Thm. 5 in~\cite{DMM08}]\\[7pt]
Multiple-regression ($\xi=2$) & $\omega/(r+1)$ &[Thm. \ref{lem:lowerS}] 
&2+$O(\sqrt{\omega/r}+\omega/r+\sqrt{k/r})$ &[Thm. \ref{lem:regressionG2a}] \\[7pt]
Multiple-regression ($\xi=\mathrm{F}$)&$1+\Omega(k/r)$ &
[Thm. \ref{lem:lowerF}] & $2+O\left(\sqrt{{k}/{r}}\right)$& [Thm. \ref{lem:regressionGF}]\\[7pt]
\hline
\end{tabular}
\end{center}
\caption{Lower bounds on the approximation ratio for different formulations of
linear regression and a coreset of size \math{r}. The
randomized algorithm in the second row of the table delivers 
a constant probability of success (all other algorithms are deterministic).
The lower bounds are values $\gamma$ such that \math{\norm{\matA\tilde\matX_{opt}-\matB}/\norm{\matA\matX_{opt}-\matB} \ge \gamma}. 
The approximation ratios are values $\beta$ such that \math{\norm{\matA\tilde\matX_{opt}-\matB}/\norm{\matA\matX_{opt}-\matB} \le \beta}.
In the first two rows in the table, $\matX_{opt}, \tilde\matX_{opt},$ and $\matB$ are vectors.
{\sc Notation:}  $n$ is the number of data points of dimension $d < n$; 
$k$ is
the rank of the matrix whose rows correspond to the $n$ data points;
\math{r} is the size of the coreset, \math{k<r<n};
 $\omega \geq 1$ is the number of ``response'' vectors
in multiple-response regression (in the last two rows in the table $\matX_{opt}, \tilde\matX_{opt},$ and $\matB$ have $\omega$ columns). 
\label{table:t2}}
\end{table}

\cite{DMM08} provides another \math{\matB}-agnostic coreset construction algorithm with $r = O(k \log k/\epsilon^2)$. For a fixed \math{\matB}, the method in \cite{DMM08} delivers a probabilistic bound on the approximation error. However, there are target matrices \math{\matB} for which the bound fails by an arbitrarily large amount. The probabilistic algorithms get away with this by brushing all these (possibly large) errors into a low probability event, with respect to random choices made in the algorithm. So, in some sense, these algorithms are not \math{\matB}-agnostic, in that they do not construct a coreset which works well for all \math{\matB} with some (say) constant probability. Nevertheless, the fact that they give a constant probability of success for a fixed but unknown \math{\matB} makes these algorithms interesting and useful. We will give a lower bound on the approximation ratio of such algorithms as well, for a given probability of success
(Theorem~\ref{lem:agnosticR}). Finally, we will give lower bounds on the size of the coreset for the general (non-agnostic) multiple regression setting (Theorems~\ref{lem:lowerS} and~\ref{lem:lowerF}).

\subsection{An Impossibility Result for \math{\matB}-Agnostic Coreset Construction}

We first present the lower bound for simple regression. Recall that a coreset construction algorithm is $\b$-agnostic if it constructs a coreset without knowledge of \math{\b}, and then provides an approximation guarantee for every \math{\b}. We show that no coreset can work for every \math{\b}; therefore a \math{\b}-agnostic coreset will be bad for some vector \math{\b}. In fact, there exists a matrix \math{\matA} such that every coreset has an associated ``bad'' \math{\b}.

\begin{theorem}[Deterministic \math{\b}-Agnostic coresets]\label{lem:agnosticD}
There exists a matrix \math{\matA\in\R^{n \times d}} such that for every coreset \math{\matC \in \R^{r \times d}} of size \math{r\le n}, there exists
\math{\b \in \R^n} (depending on \math{\matC}) for which
$$
\TNormS{\matA \tilde{\x}_{opt} - \b}
\geq \frac{n}{r}
\TNormS{\matA \x_{opt} - \b}.
$$
\end{theorem}

\begin{proof}
Let \math{\matA} be any matrix with orthonormal columns whose first column is \math{\bm1_n/\sqrt{n}}, and consider any coreset \math{\matC} of size \math{r}. Let \math{\b=\bm1_{\overline\matC}/\sqrt{n-r}}, where \math{\bm1_{\overline\matC}} is the \math{n}-vector of \math{1}'s except at the coreset locations. So for the coreset
regression, \math{\b_c=\bm 0}, and so \math{\tilde\x_{opt}=\bm0_{d \times 1}}. Therefore, 
$$\TNormS{\matA\tilde\x_{opt}-\b}=\TNormS{\b}=1.$$ Let \math{\matP_\matA} project onto the columns of
\math{\matA} and \math{\matP_{\matA^{(1)}}} project onto  the first column of \math{\matA}. The following sequence establishes the result:
\mand{
\TNormS{\matA \x_{opt} - \b}
=
\TNormS{(\matI - \matP_\matA)\b}
\le
\TNormS{(\matI - \matP_{\matA^{(1)}})\b}
=
\frac{r}{n}
}
\end{proof}

\noindent We now consider randomized algorithms that construct a coreset without looking at \math{\b} (e.g. \cite{DMM08}). These algorithms work for any fixed (but unknown) \math{\b}, and
deliver a probabilistic approximation guarantee for any single  fixed \math{\b}; in some sense they are \math{\b}-agnostic. By the previous discussion, the returned coreset must fail for some \math{\b}, i.e., the probabilistic guarantee does not hold for all \math{\b}, and, when it fails, it could do so with very bad error. We will now present a lower bound on the approximation accuracy of such existing randomized algorithms for coreset construction, even for a single \math{\b}.

First, we define randomized coreset construction algorithms. Let \math{\matC_1,\matC_2,\ldots,\matC_{\choose{n}{r}}} be the
\math{\choose{n}{r}} different coresets of size \math{r}. A randomized
algorithm assigns probabilities
\math{p_1,p_2,\ldots,p_{\choose{n}{r}}} to each coreset, and selects
one according to these probabilities. The probabilities \math{p_i}
may depend on~\math{\matA}. The algorithm is \math{\b}-agnostic if
 the probabilities \math{p_i} do not depend on \math{\b}.
As usual, let \math{r} be the size of the coreset.

\begin{theorem}[Probabilistic $\b$-Agnostic Coresets]
\label{lem:agnosticR}
For any randomized \math{\b}-agnostic coreset construction algorithm, and any integer
\math{0\le \ell\le n-r}, there exists \math{\matA\in\R^{n\times d}}
 and \math{\b\in \R^n}, such that, with probability at least
\math{\choose{n-r}{\ell}/\choose{n}{\ell}},
$$
\TNormS{\matA \tilde{\x}_{opt} - \b }  \geq \frac{n}{n-\ell}
\TNormS{\matA \x_{opt} - \b}.
$$
\end{theorem}

\begin{proof}
Let \math{\matA} be any matrix with orthonormal columns whose first
column is \math{\bm1_n/\sqrt{n}}, as in the proof of
Theorem~\ref{lem:agnosticD}.
 Let \math{\matT} be a set of size \math{\ell\le n-r}.
The neighborhood \math{N(\matT)} is the set of coresets (of size \math{r})
that have
 non-empty
intersection with \math{\matT}.
Every coreset appears in
\math{\choose{n}{\ell}-\choose{n-r}{\ell}} such neighborhoods
(the number of sets of size
\math{\ell} which intersect with a coreset of size \math{r}).
Let \math{\matC} be the random coreset (of size \math{r}) selected by the
algorithm.
Let
\math{\Probab{\matC\in N(\matT)}} 
be the probability that the coreset selected by the
algorithm is in
\math{N\left(\matT\right)}; then, 
\math{\Probab{\matC\in 
N(\matT)}=\sum_{\matC_i\in N\left(\matT\right)} \Probab{\matC_i}}.
Therefore,
\mand{
\sum_{\matT}\Probab{\matC\in N(\matT)}
=
\sum_{\matT}\sum_{\matC_i\in N\left(\matT\right)}\Probab{\matC_i}
=\choose{n}{\ell}-\choose{n-r}{\ell},
}
where the last equality follows because each coreset appears exactly
\math{\choose{n}{\ell}-\choose{n-r}{\ell}}
times in the summation and \math{\sum_i \Probab{\matC_i} = 1}.
Thus, there is at least one set \math{\matT^*} for which
\mand{
\Probab{\matC\in N(\matT^*)}\le
\frac{\choose{n}{\ell}-\choose{n-r}{\ell}}{\choose{n}{\ell}}=
1-\frac{\choose{n-r}{\ell}}{\choose{n}{\ell}}.
}
So, with probability at least \math{\choose{n-r}{\ell}/\choose{n}{\ell}},
the selected coreset does not intersect with \math{\matT^*}. Select
\math{\b=\bm1_{\matT^*}} (the unit vector which is \math{1/\sqrt{\ell}}
 at the indices
corresponding to \math{\matT^*}).
Now, with probability at least
\math{\choose{n-r}{\ell}/\choose{n}{\ell}},
\math{\tilde\x_{opt}=\bm0}, and the
analysis in the proof of Theorem \ref{lem:agnosticD} shows that
\mand{\TNormS{\matA\tilde\x_{opt}-b}\ge \frac{n}{n-\ell}
\TNormS{\matA\x_{opt}-b}.}
\end{proof}

By Stirling's formula, after some
algebra, the probability \math{\choose{n-r}{\ell}/\choose{n}{\ell}} is
asymptotic
 to~\math{e^{-2r\ell/n}}. Setting \math{\ell=\Theta(n/r)} gives a success probability that is a constant. Then, the approximation ratio cannot be better than \math{1+\Omega(1/r)}. With regard to high probability (approaching one) algorithms, consider \math{\ell=n\log n /2r} to conclude  that if the success probability is at least \math{1-1/n}, the approximation ratio is no better than \math{1+\log(n)/(2r-\log n)}.

\subsection{Lower Bounds for Non-Agnostic Multiple Regression}
\label{sec:lowerG}

For both the spectral and the Frobenius norm,
we now consider non-agnostic unconstrained multiple regression, and
give lower bounds for coresets of size $r > d=\rank(\matA)$ (for simplicity,
we set \math{\rank(\matA)=d}). The results are presented
in Theorems~\ref{lem:lowerS} and~\ref{lem:lowerF}.

\begin{theorem}[Spectral Norm]
\label{lem:lowerS}
There exists $\matA\in\R^{n \times d}$ with rank $d$ and $\matB\in\R^{n \times \omega}$
such that for any
$r >d$ and
any sampling and rescaling matrices
 $\matS \in \R^{r \times n}$ and  $\matD \in \R^{r \times r}$,
the solution to the coreset regression
$\tilde{\matX}_{opt}= (\matD \matS\matA)^{\dagger} \matD \matS \matB \in \R^{d \times \omega}$ satisfies
$$ \norm{\matA \tilde{\matX}_{opt} - \matB} _2^2
\geq \frac{\omega}{r+1}  \TNormS{\matA \matX_{opt} - \matB} .$$
\end{theorem}

\begin{proof}

First, we need some results from \cite{BDM11a}. Consider the matrix
$$\matH = [\e_1+\alpha\e_2, \e_1+\alpha\e_3,\ldots, \e_1+\alpha\e_{\omega}] \in\R^{\omega\times (\omega-1)},$$
where \math{\e_i\in\R^{\omega}} are the standard basis vectors. Then, let $\matB = \matH\transp \in \R^{(\omega-1) \times \omega}$.
Theorem 34 in \cite{BDM11a} (with $\alpha=1$) argues the following: given $\matB$ and any sampling matrix \math{\matS \in \R^{r \times (\omega-1)}} and diagonal rescaling matrix
$\matD \in \R^{r \times r}$,
with \math{\hat\matC=\matD\matS\matB} (rescaled sampled coreset of \math{\matB}), and any $k$ with $1\le k \le \omega-1$,
$$ \TNormS{\matB - \Pi_{\hat\matC,k}(\matB)} \geq  \frac{\omega}{r+1}
\TNormS{\matB - \matB_k}.$$
In the above, \math{\Pi_{\hat\matC,k}(\matB) \in \R^{(\omega-1) \times \omega}} of rank $k$
is the best rank-\math{k}
approximation to \math{\matB} (in the spectral norm) whose rows lie in the span of all the rows
in \math{\hat\matC} (the row-space of \math{\hat\matC}); and,
\math{\matB_k  \in \R^{(\omega-1) \times \omega}} of rank $k$ is the best rank-\math{k} approximation to
\math{\matB} (which could be computed via the truncated SVD of
\math{\matB}).\footnote{Actually, \math{\matD} is irrelevant here because the
row-space of \math{\matS\matB} is the same as 
the row space of \math{\matD \matS \matB}.}

Since $\Pi_{\hat\matC,k}(\matB)$  is the best rank-\math{k}
approximation to \math{\matB} in the row-space of \math{\hat\matC},
it follows that $$\TNormS{\matB - \Pi_{\hat\matC,k}(\matB)}
\le \TNormS{\matB - \matX\hat\matC},$$ for  any
\math{\math\matX \in
\R^{(\omega-1)\times r}} with rank at most \math{k}
(because \math{\matX\hat\matC} will have rank at most \math{k} and is in the
row space of \math{\hat\matC}).
Set \math{\matX= \matU_{\matB,k} (\matD\matS \matU_{\matB,k})^{\dagger}}, where
\math{\matU_{\matB,k}\in\R^{(\omega-1)\times k}}
has \math{k} columns which are the top-\math{k} left singular vectors of
\math{\matB}.
It is easy to verify that
\math{\matX} has the correct dimensions and rank at most \math{k}.
Since \math{\hat\matC=\matD\matS\matB}, we have that
$$ \TNormS{\matB - \Pi_{\hat\matC,k}(\matB)}
\leq \TNormS{\matB-\matU_{\matB,k} (\matD\matS \matU_{\matB,k})^{\dagger} \matD
\matS \matB}.$$

We now  construct the regression problem which exhibits the lower bound in the theorem.
Let \math{\matA=\matU_{\matB,d}\in\R^{(\omega-1)\times d}} (i.e., we choose \math{k=d} in the above discussion) and
\math{n=\omega-1}. $\matB$ is as we described above.
Suppose a coreset construction algorithm gives sampling and rescaling matrices
\math{\matS} and \math{\matD}, for a coreset of size \math{r}.
So, the coreset regression is with
\math{\tilde\matA=\matD\matS\matA=\matC} and \math{\tilde\matB=\matD\matS\matB}.
The solution to the coreset regression is
\mand{\tilde\matX_{opt} = \matC^{\dagger}\matD\matS\matB=
\mat(\matD\matS\matA)^{\dagger}\matD\matS\matB
=
\mat(\matD\matS\matU_{\matB,d})^{\dagger}\matD\matS\matB,
}
which means that
\mand{
\norm{\matA\tilde\matX_{opt}-\matB}_2^2
=
\TNormS{
\matU_{\matB,d}\mat(\matD\matS\matU_{\matB,d})^{\dagger}\matD\matS\matB-\matB
}
\ge
\TNormS{\Pi_{\matC,d}(\matB)-\matB }
\ge
\frac{\omega}{r+1}
\TNormS{\matB_d - \matB}.
}
To conclude the proof, observe that
\math{\matB_d=\matU_{\matB,d}\matU_{\matB,d}\transp
\matB=\matA\matA^{\dagger}\matB=\matA\matX_{opt}}.
\end{proof}

\begin{theorem} [Frobenius Norm]\label{lem:lowerF}
There exists $\matA\in\R^{n \times d}$ of rank $d$ and $\matB\in\R^{n \times \omega}$
such that for any
$r >d$ and
any sampling and rescaling matrices
 $\matS \in \R^{n \times r}$ and  $\matD \in \R^{r \times r}$,
the solution to the coreset regression
$\tilde{\matX}_{opt}= (\matD \matS \matA)^{\dagger} \matD \matS \matB \in \R^{d \times \omega}$ satisfies (for any $\alpha > 0$)
$$ \FNormS{\matA \tilde{\matX}_{opt} - \matB}
\geq  \frac{n-r}{n-d}\left(1+\frac{d}{r+\alpha^2}\right) \FNormS{\matA \matX_{opt} - \matB}.$$
As $\alpha \rightarrow 0$ and \math{n\rightarrow\infty} the  lower bound is $1+d/r$.
\end{theorem}

\begin{proof}
First, we need some results from \cite{BDM11a}. For any integer $\gamma > 1$ and any integer $k \ge 1$, Theorem 36 in
\cite{BDM11a} exhibits a matrix
\math{\matB \in \R^{\gamma k \times(\gamma+1)k}} such that for any
sampling matrix
\math{\matS \in \R^{r \times \gamma k}} and diagonal rescaling matrix
 $\matD \in \R^{r \times r}$,
with \math{\hat\matC=\matD\matS\matB} (rescaled sampled coreset of \math{\matB}),
any \math{\alpha>0}, and any $r \ge 1$,
$$\frac{\FNormS{\matB - \Pi_{\hat\matC,k}(\matB)}}{\FNorm{\matB-\matB_k}^2}
\ge\frac{\gamma k-r}{\gamma k-k}\left(1+\frac{k}{r+\alpha^2}\right).
$$
The matrix $\matB$ is constructed as follows. Recall that $\gamma$ is any positive integer with $\gamma > 1$.
Let $\matA$ have dimensions $(\gamma+1) \times \gamma$ and be constructed as follows.
$$\matA = 
\left[\e_1+ \frac{\alpha}{\sqrt{k}}\e_2, \e_1+ \frac{\alpha}{\sqrt{k}}\e_3,\ldots, \e_1+ \frac{\alpha}{\sqrt{k}}\e_{\gamma}\right],$$
where \math{\e_i\in\R^{\gamma+1}} are the standard basis vectors.
Now construct $\matH$ to be block diagonal, with $k$ copies of $\matA$ along its diagonal;
so, the dimensions of \math{\matH} are  $(\gamma+1)k \times \gamma k$. Then, $\matB = \matH\transp$.

In the above, \math{\Pi_{\hat\matC,k}(\matB)\in \R^{\gamma k \times(\gamma+1)k}} of rank $k$
is the best rank-\math{k}
approximation to \math{\matB} (in the Frobenius norm) whose rows lie in the span of all the rows
in \math{\hat\matC} (the row-space of \math{\hat\matC}); and,
\math{\matB_k  \in\R^{\gamma k \times(\gamma+1)k}} of rank $k$ is the best rank-\math{k} approximation to
\math{\matB} (which could be computed via the truncated SVD of
\math{\matB}).
Since $\Pi_{\hat\matC,k}(\matB)$  is the best rank-\math{k}
approximation to \math{\matB} in the row-space of \math{\hat\matC},
it follows that $$\FNormS{\matB - \Pi_{\hat\matC,k}(\matB)}
\le \FNormS{\matB - \matX\hat\matC},$$ for  any
\math{\math\matX \in
\R^{\gamma k \times r}} with rank at most \math{k}
(because \math{\matX\hat\matC} will have rank at most \math{k} and is in the
row space of \math{\hat\matC}).
Set \math{\matX= \matU_{\matB,k} (\matD\matS \matU_{\matB,k})^{\dagger}}, where
\math{\matU_{\matB,k}\in\R^{\gamma k  \times k}}
has \math{k} columns which are the top-\math{k} left singular vectors of
\math{\matB}.
It is easy to verify that
\math{\matX} has the correct dimensions and rank at most \math{k}.
Since \math{\hat\matC=\matD\matS\matB}, we have that
$$ \FNormS{\matB - \Pi_{\hat\matC,k}(\matB)}
\leq \FNormS{\matB-\matU_{\matB,k} (\matD\matS \matU_{\matB,k})^{\dagger} \matD
\matS \matB}.$$

We now  construct the regression problem which proves the lower bound in the theorem.
Let \math{\matA=\matU_{\matB,d}\in\R^{\gamma d \times d}} (i.e., we choose \math{k=d} in the above discussion),
\math{n=\gamma d} (i.e. $n$ is a multiple of $d$ in the regression problem), and $\omega = (\gamma+1)d$. $\matB$ is as we described above.
Suppose a coreset construction algorithm gives sampling and rescaling matrices
\math{\matS} and \math{\matD}, for a coreset of size \math{r > d}.
So, the coreset regression is with
\math{\matC=\matD\matS\matA \in \R^{r \times d}} and \math{\matD\matS\matB \in \R^{r \times \omega}}.
The solution to the coreset regression is
\mand{\tilde\matX_{opt} = \matC^{\dagger}\matD\matS\matB=
\mat(\matD\matS\matA)^{\dagger}\matD\matS\matB
=
\mat(\matD\matS\matU_{\matB,d})^{\dagger}\matD\matS\matB,
}
which means that
\mand{
\norm{\matA\tilde\matX_{opt}-\matB}_{\mathrm{F}}^2
=
\FNormS{
\matU_{\matB,d}\mat(\matD\matS\matU_{\matB,d})^{\dagger}\matD\matS\matB-\matB
}
\ge
\FNormS{\Pi_{\matC,d}(\matB)-\matB }
\ge
\frac{\omega-d-r}{\omega-2d}\left(1+\frac{d}{r+\alpha^2}\right)
\FNormS{\matB_d - \matB}.
}
To conclude the proof, observe that
\math{\matB_d=\matU_{\matB,d}\matU_{\matB,d}\transp
\matB=\matA\matA^{\dagger}\matB=\matA\matX_{opt}} and $\omega = n+d$.
\end{proof}

\section{Open problems}

An important open problem arises in our work: can we determine the minimum size of a coreset that provides a \math{(1+\epsilon)} relative-error guarantee for simple linear regression? We conjecture that $\Omega\left(k/ \epsilon\right)$ is a lower bound, which will make our results almost tight. Certainly, coresets of size exactly \math{k} cannot be guaranteed: consider two data points \math{(1,1),(-1,1)}. The optimal regression is zero; however any coreset of size one will give non-zero regression.

\paragraph{Acknowledgements.} Christos Boutsidis acknowledges the support from XDATA program of the Defense Advanced Research Projects Agency (DARPA), administered through Air Force Research Laboratory contract FA8750-12-C-0323. Petros Drineas and Malik Magdon-Ismail have been supported by NSF CCF 1016501, NSF DMS 1008983, and NSF CCF CAREER 824684.

\bibliographystyle{plain}
\bibliography{RPI_BIB}

\clearpage

\appendix

\section{Linear Algebra Background}\label{sec:secBack}

The Singular Value Decomposition (SVD) of a matrix $\matA \in \R^{n \times d}$ of rank $k$ is
a decomposition $$\matA = \matU_\matA \matSig_\matA \matV_\matA\transp.$$ The singular values \math{\sigma_1\ge\sigma_2\ge\cdots\ge\sigma_k>0} are contained in the diagonal matrix \math{\matSig_\matA \in \R^{k \times k}}; $\matU_\matA \in\R^{n \times k}$ contains the left singular vectors of $\matA$; and $\matV_{\matA} \in \R^{d \times k}$ contains the right singular vectors. 
The Moore-Penrose pseudo-inverse of $\matA$ is
$\matA^{\dagger} = \matV_{\matA} \matSig_\matA^{-1} \matU_{\matA}\transp.$
Given an orthonormal matrix $\matU_\matA \in\R^{n \times k}$, the perpendicular matrix $\matU_\matA^\perp\in\R^{n \times (n-k)}$ to $\matU_{\matA}$ satisfies: 
$(\matU_\matA^\perp)\transp \matU_\matA^\perp = \matI_{n-k}$, $\matU_\matA\transp \matU_\matA^\perp = \bm{0}_{k \times (n-k)}$, and $\matU_\matA \matU_{\matA}\transp + \matU_\matA^\perp (\matU_\matA^\perp)\transp= \matI_{n}$. All the singular values of both $\matU_\matA$ and $\matU_\matA^\perp$ are equal to one. Given $\matU_{\matA}$, $\matU_\matA^\perp$ can be computed in deterministic $O\left(n \left(n-k\right)^2\right)$ time via the QR factorization.

We remind the reader of the Frobenius and spectral matrix norms: $ \FNormS{\matA} = \sum_{i,j}\matA_{ij}^2 =\sum_{i=1}^{k}\sigma_i^2 $ and $\TNorm{\matA}^2 = \sigma_1^2$. We will sometimes use the notation $\XNorm{\matA}$ to indicate that an expression holds for both $\xi = 2$ or $\xi = \mathrm{F}$.
For any two matrices $\matX$ and $\matY$,
$
\TNorm{\matX}\le\FNorm{\matX}\le\sqrt{\rank(\matX)}\TNorm{\matX}; \qquad \FNorm{\matX\matY} \leq \FNorm{\matX}\TNorm{\matY}; \qquad \FNorm{\matX\matY} \leq \TNorm{\matX}\FNorm{\matY}.
$
These are stronger variants of the standard submultiplicativity property $\XNorm{\matX\matY} \leq \XNorm{\matX} \XNorm{\matY}$ and we will refer to them as spectral submultiplicativity. It follows that, if $\matQ$ is orthonormal, then $\XNorm{\matQ \matX} \leq \XNorm{\matX}$ and $\XNorm{\matY \matQ\transp } \leq \XNorm{\matY}$. Finally, we will make frequent use of the following two lemmas.

\begin{lemma}[matrix-Pythagoras]\label{lem:pythagoras}
Let \math{\matX} and \math{\matY} be two $n \times d$ matrices. If \math{\matX\matY\transp=\bm{0}_{n \times n}} or \math{\matX\transp\matY=\bm{0}_{d \times d}}, then
$$\XNorm{\matX+\matY}^2 \leq \XNorm{\matX}^2+\XNorm{\matY}^2.$$
\end{lemma}

\begin{lemma}[Fact 6.4.12 in~\cite{Bernstein05}]
\label{lem:pseudo}
Let $\matA \in \R^{m \times n}, \matB \in \R^{n \times \ell}$, and assume that $\rank(\matA) = \rank(\matB) = n$. Then, $$(\matA\matB)^{\dagger}=\matB^{\dagger}\matA^{\dagger}.$$
\end{lemma}

\section{Algorithms and Proofs of Lemmas~\ref{lem:2setS} and~\ref{lem:2setF}}
We now provide all the details of the proofs and the corresponding algorithms of 
Lemmas~\ref{lem:2setS} and~\ref{lem:2setF}. Those results, which have been described in detail in~\cite{BoutsidisPhD}, are slight extensions
of two algorithms presented in~\cite{BDM11a}, which themselves extend the original spectral sparsification result of Batson, Spielman, and Srivastava~\cite{BSS09}. More specifically, Lemma~\ref{theorem:3setGeneralS} below - in some sense - generalizes Lemma~\ref{lem:1set}; indeed,
setting $\matV = \matQ\transp := \matU$ in  Lemma~\ref{theorem:3setGeneralS}  gives Lemma~\ref{lem:1set}. Lemma~\ref{theorem:3setGeneralF} below also describes a deterministic algorithm for sampling columns from two matrices but the goal here is to optimize different spectral properties in the sampled matrices. 

In this section of the Appendix, we will slightly abuse notation by denoting with $\hat\matS \in \R^{n \times r}$ a sampling matrix which samples columns - not rows - from matrices.
We will later use $\matS = \hat\matS\transp$ to be consistent with the notation used throughout the paper. 
\begin{algorithm}[t] \label{alg:sample1}
\begin{framed}
   \caption{{\sf DeterministicSamplingI} (Lemma~\ref{theorem:3setGeneralF})}
   {\bfseries Input:} $\matV\transp=[\v_1, \v_2, \ldots,\v_n] \in \R^{k \times n}$,
\math{\matB=[\b_1, \b_2, \ldots,\b_n] \in \R^{\ell_1 \times n}},
and \math{r > k}. \\
    {\bfseries Output:} Sampling matrix $\hat\matS \in \R^{n \times r}$ and rescaling matrix $\matD \in \R^{r \times r}$.
 \begin{algorithmic}[1]
    \STATE Initialize  $\matA_0 = \bm{0}_{k \times k}$, $\hat\matS = \bm{0}_{n \times r}$, and $\matD=\bm{0}_{r \times r}$.
    \STATE Set constants $
\delta_\matB=\FNormS{\matB} (1-\sqrt{k/r})^{-1};\
\delta_L=1$.
   \FOR{$\tau=0$ {\bfseries to} $r-1$}
   \STATE Let \math{\scl_\tau=\tau-\sqrt{rk}}.
   \STATE Pick index $i_{\tau}\in\{1,2,...,n\}$ and number $t_{\tau}>0$ (see text for the definition of $U, L$):
    $$\hspace*{-0.2in} U(\b_{i_{\tau}},\delta_\matB) \le\frac{1}{t_{\tau}}\le
       L(\v_{i_{\tau}}, \delta_L, \matA_{\tau},\scl_\tau).$$
   \STATE Update $\matA_{\tau+1} = \matA_{\tau} + t_{\tau} \v_{i_{\tau}} \v_{i_{\tau}}\transp$; set $\hat\matS_{{i_{\tau}},\tau+1} = 1$ and $\matD_{\tau+1,\tau+1} = 1/\sqrt{t_{\tau}}$.
   \ENDFOR
   \STATE Multiply all the weights in $\matD$ by
\mand{\sqrt{ r^{-1}(1-\sqrt{k/r}) }.}
   \STATE {\bfseries Return:} $\hat\matS$ and $\matD$.
\end{algorithmic}
\end{framed}
\end{algorithm}
\begin{algorithm}[t] \label{alg:sample2}
\begin{framed}
   \caption{{\sf DeterministicSamplingII} (Lemma~\ref{theorem:3setGeneralS})}
   {\bfseries Input:} $\matV\transp=[\v_1, \v_2, \ldots,\v_n] \in \R^{k \times n}$,  \math{\matQ=[\q_1, \q_2,\ldots,\q_d] \in \R^{\ell_2 \times n}},
and \math{r > k}. \\
    {\bfseries Output:} Sampling matrix $\hat\matS \in \R^{n \times r}$ and rescaling matrix $\matD \in \R^{r \times r}$.
 \begin{algorithmic}[1]
    \STATE Initialize  $\matA_0 = \bm{0}_{k \times k}$, $\matB_0 = \bm{0}_{\ell_2 \times \ell_2}$, $\matOmega = \bm{0}_{n \times r}$, and $\hat\matS=\bm{0}_{r \times r}$.
    \STATE Set constants $
\delta_\matQ=\left( 1 + \ell_2/r \right) \left(1-\sqrt{k/r}\right)^{-1};\ \delta_L=1$.
   \FOR{$\tau=0$ {\bfseries to} $r-1$}
   \STATE Let \math{\scl_\tau=\tau-\sqrt{rk}}; 
   \math{\scu_{\tau} = \delta_{\matQ} \left( \tau + \sqrt{\ell_2 r} \right) }
   \STATE Pick index $i_{\tau}\in\{1,2,...,n\}$ and number $t_{\tau}>0$ (see text for the definition of $U, L$):
    $$\hspace*{-0.2in} \hat{U}(\q_{i_{\tau}},\delta_\matQ,\matB_{\tau},\scu_\tau) \le\frac{1}{t_{\tau}}\le
       L(\v_{i_{\tau}}, \delta_L, \matA_{\tau},\scl_\tau).$$
   \STATE Update $\matA_{\tau+1} = \matA_{\tau} + t_{\tau} \v_{i_{\tau}} \v_{i_{\tau}}\transp$;
                 $\matB_{\tau+1} = \matB_{\tau} + t_{\tau} \q_{i_{\tau}} \q_{i_{\tau}}\transp$, and \\
   set $\hat\matS_{{i_{\tau}},\tau+1} = 1$, $\matD_{\tau+1,\tau+1} = 1/\sqrt{t_{\tau}}$.
   \ENDFOR
   \STATE Multiply all the weights in $\matD$ by $ \sqrt{ r^{-1} \left(1-\sqrt{k/r}\right) }.$
   \STATE {\bfseries Return:} $\hat\matS$ and $\matD$.
\end{algorithmic}
\end{framed}
\end{algorithm}
\begin{lemma} [Lemma 13 in~\cite{BDM11a}]
\label{theorem:3setGeneralF}
Let \math{\matV\transp \in \R^{k \times n}} and
\math{\matB \in \R^{\ell_1 \times n}} with \math{\matV\transp\matV=\matI_k}. Let $r > k$.
Algorithm~\ref{alg:sample1} runs in $O(rk^2n + \ell_1n )$ time and deterministically constructs a sampling matrix \math{\hat\matS\in\R^{n \times r}} and
a rescaling matrix \math{\matD\in \R^{r \times r}} such that,
\begin{align*}
&\sigma_k(\matV\transp \hat\matS \matD) \ge 1 - \sqrt{{k}/{r}};
&\FNorm{\matB \hat\matS \matD}         \le   \FNorm{\matB}.
\end{align*}
We write $[\matD, \hat\matS] = DeterministicSamplingI(\matV\transp, \matB, r)$ to denote this procedure.
\end{lemma}

Algorithm \ref{alg:sample1} is a greedy technique that selects columns one at a time.
To describe the algorithm in more detail, it is convenient to view the input matrices as two sets of \math{n} vectors,
$$\matV\transp=[\v_1, \v_2, \ldots,\v_n],$$ and
$$\matB=[\b_1, \b_2, \ldots,\b_n].$$
Given  $k$ and $r>k$, introduce the
iterator $\tau = 0, 1,2,...,r-1,$ and define the
parameter $$\scl_\tau=\tau-\sqrt{rk}.$$
For a square symmetric
matrix \math{\matA\in\R^{k\times k}} with eigenvalues \math{\lambda_1,\ldots,
\lambda_k}, $\v \in \R^k$ and $\scl\in\R$, define
$$
\phi(\scl, \matA) =
\sum_{i=1}^k\frac{1}{\lambda_i-\scl},
$$
and let
$L(\v, \delta_L, \matA, \scl)$  be defined as
$$ L(\v, \delta_L, \matA, \scl)  =
\frac{\v\transp(\matA-\scl'\matI_k)^{-2}\v}
{\phi(\scl', \matA)-\phi(\scl,\matA)} -\v\transp(\matA-\scl'\matI_k)^{-1}\v,$$
where $\scl'= \scl+\delta_L = \scl+1.$
For a vector $\z$ and scalar \math{\delta > 0}, define the function
$$
U(\z,\delta) = \frac{1}{\delta} \z\transp \z .
$$
At each iteration \math{\tau}, the algorithm selects
$i_{\tau}$, $t_{\tau} > 0$
for which
$$U(\b_{i_{\tau}},\delta_\matB)\le t_{\tau}^{-1} \le
       L(\v_{i_{\tau}},\delta_L,\matA_{\tau},\scl_\tau).$$
The running time of the algorithm
is dominated by the search for an index $i_{\tau}$
satisfying $$U(\b_{i_{\tau}},\delta_\matB)\le t_{\tau}^{-1}\le
       L(\v_{i_{\tau}},\delta^{-1},\matA_{\tau},\scl_\tau)$$ (one can achieve that by exhaustive search). One needs
\math{\phi(\scl,\matA)}, and hence the eigenvalues of
\math{\matA}. This takes
\math{O(k^3)} time, once per iteration, for a total of
\math{O(rk^3)}. Then, for \math{i=1,\ldots,n}, we need to compute
\math{L} for every \math{\v_i}. This takes \math{O(nk^2)} per iteration,
for a total of \math{O(rnk^2)}. To compute \math{U}, we need
\math{\b_i\transp\b_i}  for \math{i=1,\ldots,n},
which need to be computed only once for the whole algorithm and
takes \math{O(\ell_1 n )}.
So, the total running time is \math{O(nrk^2+\ell_1 n)}.
\begin{lemma} [Lemma 12 in~\cite{BDM11a}]
\label{theorem:3setGeneralS}
Let \math{\matV\transp \in \R^{k \times n}}, \math{\matQ \in
\R^{\ell_2 \times n} }, \math{\matV\transp\matV=\matI_k}, and
\math{\matQ\transp\matQ=\matI_{\ell_2}}. Let $r > k$.
Algorithm~\ref{alg:sample2} runs in  $O( r k^2 n + r \ell_2^2n )$ time and deterministically constructs a sampling matrix
\math{\hat\matS\in\R^{n \times r}} and
a rescaling matrix \math{\matD\in \R^{r \times r}} such that,
\begin{align*}
&\sigma_k(\matV\transp \hat\matS \matD) \ge 1 - \sqrt{{k}/{r}};
&\TNorm{\matQ \hat\matS \matD}          \le  1+\sqrt{\ell_2/r}.
\end{align*}
If $\matQ = \matI_{n}$, it runs in $O( r k^2 n )$;
we write $[\matD, \hat\matS] = DeterministicSamplingII(\matV\transp, \matQ, r)$ for this procedure.
\end{lemma}
Algorithm \ref{alg:sample2} is similar to Algorithm \ref{alg:sample1}; we only need to define the function
$\hat{U}$.
For a square symmetric matrix \math{\matB\in\R^{\ell_2 \times \ell_2}} with eigenvalues \math{\lambda_1,\ldots,
\lambda_{\ell_2}}, $\q \in \R^{\ell_2}$, $\scu\in\R$, define:
$
\hat{\phi}(\scu,\matB) =
\sum_{i=1}^{\ell_2}\frac{1}{\scu - \lambda_i},
$
and let
$\hat{U}(\q, \delta_{\matQ}, \matB, \scu)$  be defined as
$ \hat{U}(\q, \delta_{\matQ}, \matB, \scu)  =
\frac{\q\transp(\matB-\scu'\matI_{\ell_2})^{-2}\q}
{\hat\phi(\scu,\matB)-\hat{\phi}(\scu', \matB) }
-
\q\transp(\matB-\scu'\matI_{\ell_2})^{-1}\q,$
where
$ \scu' = \scu + \delta_{\matQ} =
\scu + \left( 1 + \ell_2/r \right) \left(1-\sqrt{k/r}\right)^{-1}
.$
The running time of the algorithm is \math{O(nrk^2+ nr\ell_2^2)}.

\subsection{Proof of Lemma~\ref{lem:2setS}}
We first restate the lemma. 
\begin{lemma} [Restatement of Lemma~\ref{lem:2setS}] 
Let \math{\matY\in\R^{n\times \ell_1}} and \math{\matPsi\in\R^{n\times \ell_2}} with respective ranks \math{\rho_{\matY}}, and \math{\rho_\matPsi}. Given \math{r>\rho_\matY}, there exists a deterministic algorithm that runs in time $T_{SVD}\left(\matY\right) + T_{SVD}\left(\matPsi\right) + O(r n (\rho_{\matY}^2+\rho_{\matPsi}^2))$ and constructs sampling and rescaling matrices \math{\matS\in\R^{r\times n}},
\math{\matD\in\R^{r\times r}} satisfying:
\mand{
\rank\left(\matD\matS \matY\right) = \rank\left(\matY\right);
\
\TNorm{\left(\matD\matS\matY\right)^{\dagger}} < \frac{1}{1 - \sqrt{{\rho_{\matY}}/{r} }}
 \TNorm{\matY^{\dagger}};
\
\TNorm{\matD\matS\matPsi}
<\left( 1 + \sqrt{ \frac{\rho_{\matPsi}}{r} }\right) \TNorm{\matPsi}.
}
If $\matPsi = \matI_{n}$, the running time of the algorithm reduces to $T_{SVD}\left(\matY\right)+O\left(r n \rho_{\matY}^2\right)$. We write $\left[\matD, \matS\right] = MultipleSpectralSampling\left(\matY,\matPsi, r\right)$ to denote such a deterministic procedure.
\end{lemma}

\begin{proof}
Let the SVD of $\matY \in \R^{n \times \ell_1}$ is $\matY = \matU_{\matY} \matSig_{\matY} \matV_{\matY}\transp$,
with $\matU_{\matY} \in \R^{n \times \rho_\matX}$, $\matSig_{\matY} \in \R^{\rho_\matY \times \rho_\matY}$, $\matV_{\matY} \in \R^{\ell_1 \times \rho_\matX}$. 
Let the SVD of $\matPsi \in \R^{n \times \ell_2}$ is $\matPsi = \matU_{\matPsi} \matSig_{\matPsi} \matV_{\matPsi}\transp$,
with $\matU_{\matPsi} \in \R^{n \times \rho_\matPsi}$, $\matSig_{\matPsi} \in \R^{\rho_\matPsi \times \rho_\matPsi}$,
and $\matV_{\matPsi} \in \R^{\ell_2 \times \rho_\matPsi }$. 
Let 
$$[\matD, \hat\matS] = DeterministicSamplingII(\matU_\matY\transp, \matU_\matPsi\transp,r).$$
By Lemma~\ref{theorem:3setGeneralS},
$$
\sigma_{\min}(\matU_\matY\transp\hat\matS\matD )\ge \left(1-\sqrt{\rho_{\matY}/{r}} \right),
$$ which implies 
$$\TNorm{(\matU_\matY\transp \hat\matS \matD )^{\dagger}}\le
\left(1-\sqrt{\rho_{\matY}/{r}} \right)^{-1},$$ 
and
$$\rank(\matU_\matY\transp\hat\matS\matD)=\rho_\matY.$$ 
Also,
$$\TNorm{\matU_\matPsi\transp \hat\matS \matD } \le \left(1+\sqrt{\rho_{\matPsi}/{r}} \right)$$
because
$$
\sigma_{\max}
(\matU_\matPsi\transp\hat\matS\matD )\le \left(1+\sqrt{\rho_{\matPsi}/{r}} \right).
$$ 
Thus,
\eqan{
\TNorm{(\matY\transp \hat\matS\matD )^{\dagger}}
&=&
\TNorm{(\matV_{\matY} \matSig_{\matY} \matU_{\matY}\transp \hat\matS \matD)^{\dagger} } \\
&{\buildrel (a)\over =}&
\TNorm{(\matU_{\matY}\transp \hat\matS \matD)^{\dagger} (\matV_{\matY} \matSig_{\matY})^{\dagger}}\\
&\le&
\TNorm{(\matU_{\matY}\transp \hat\matS \matD)^{\dagger}}
\TNorm{ (\matV_{\matY} \matSig_{\matY})^{\dagger}}\\
&\le&
\left(1-\sqrt{\rho_{\matY}/{r}} \right)^{-1}
\TNorm{ (\matV_{\matY} \matSig_{\matY})^{\dagger}}\\
&=& \left(1-\sqrt{\rho_{\matY}/{r}} \right)^{-1}
\TNorm{(\matY\transp)^{\dagger}}
}
(a) uses Lemma ~\ref{lem:pseudo}. To obtain the first inequality in the lemma we need to take $\matS = \hat\matS\transp$
and observe that  $\TNorm{(\matY\transp \hat\matS\matD )^{\dagger}} = \TNorm{( \matD \matS \matY )^{\dagger}}$, and
$\TNorm{(\matY\transp)^{\dagger}} = \TNorm{\matY^{\dagger}}$. 
We now prove the second inequality in the lemma,
$$
\TNorm{ \matPsi \transp \hat\matS \matD }  = 
 \TNorm{ \matV_\matPsi \matSig_\matPsi \matU_\matPsi\transp \hat\matS\matD } \leq 
 \TNorm{ \matV_\matPsi \matSig_\matPsi }  \TNorm{ \matU_\matPsi\transp \hat\matS\matD } = 
\TNorm{\matPsi\transp}   \TNorm{ \matU_\matPsi\transp \hat\matS\matD } \leq  \TNorm{\matPsi\transp}    \left(1+\sqrt{\rho_{\matPsi}/{r}} \right).
$$
To obtain the second inequality in the lemma we need to take $\matS = \hat\matS\transp$
and use  $  \TNorm{ \matPsi \transp \hat\matS \matD } = \TNorm{ \matD \matS \matPsi  },$ and
$\TNorm{\matPsi\transp} = \TNorm{\matPsi}$. 
\end{proof}

\subsection{Proof of Lemma~\ref{lem:2setF}}
We first restate the lemma. 
\begin{lemma} [Restatement of Lemma~\ref{lem:2setF}]
Let \math{\matY\in\R^{n\times \ell_1}} and \math{\matPsi\in\R^{n\times \ell_2}} with respective ranks \math{\rho_{\matY}}, and \math{\rho_\matPsi}. Given \math{r>\rho_\matY}, there exists a deterministic algorithm that runs in time $T_{SVD}(\matY)+O(r n \rho_{\matY}^2  + \ell_2 n)$ and constructs sampling and rescaling matrices \math{\matS\in\R^{r\times n}},
\math{\matD\in\R^{r\times r}} satisfying:
\mand{
\rank\left(\matD\matS \matY\right) = \rank\left(\matY\right);
\
\TNorm{\left(\matD\matS\matY\right)^{\dagger}} < \frac{1}{1 - \sqrt{{\rho_{\matY}}/{r}}}
 \TNorm{\matY^{\dagger}};  \quad
\
\FNorm{\matD\matS\matPsi} \le \FNorm{\matPsi}.
}
If $\matPsi = \matI_{n}$, the running time of the algorithm reduces to $T_{SVD}\left(\matY\right)+O\left(r n \rho_{\matY}^2\right)$. We write $\left[\matD, \matS\right] = MultipleFrobeniusSampling\left(\matY, \matPsi, r\right)$ to denote such a deterministic procedure.
\end{lemma}

\begin{proof}
Let the SVD of $\matY \in \R^{n \times \ell_1}$ is $\matY = \matU_{\matY} \matSig_{\matY} \matV_{\matY}\transp$,
with $\matU_{\matY} \in \R^{n \times \rho_\matX}$, $\matSig_{\matY} \in \R^{\rho_\matY \times \rho_\matY}$, $\matV_{\matY} \in \R^{\ell_1 \times \rho_\matX}$. 
Let the SVD of $\matPsi \in \R^{n \times \ell_2}$ is $\matPsi = \matU_{\matPsi} \matSig_{\matPsi} \matV_{\matPsi}\transp$,
with $\matU_{\matPsi} \in \R^{n \times \rho_\matPsi}$, $\matSig_{\matPsi} \in \R^{\rho_\matPsi \times \rho_\matPsi}$,
and $\matV_{\matPsi} \in \R^{\ell_2 \times \rho_\matPsi }$. 
Let 
$$[\matD,\hat\matS] = DeterministicSamplingI(\matU_\matY\transp,\matPsi\transp,r).$$
By Lemma~\ref{theorem:3setGeneralF},
$$
\sigma_{\min}(\matU_\matY\transp\hat\matS\matD )\ge \left(1-\sqrt{\rho_{\matY}/{r}} \right),
$$ which implies 
$$\TNorm{(\matU_\matY\transp \hat\matS \matD )^{\dagger}}\le
\left(1-\sqrt{\rho_{\matY}}/{r} \right)^{-1},$$ 
and
$$\rank(\matU_\matY\transp\hat\matS\matD)=\rho_\matY.$$ 
Also,
$$\FNorm{\matPsi\transp \hat\matS \matD } \le \FNorm{\matPsi\transp },$$
which by taking $\matS = \hat\matS\transp$ gives the second inequality in the lemma,
$$\FNorm{\matD \matS \matPsi } \le \FNorm{\matPsi}.$$ 
Now we prove the first inequality in the lemma,
\eqan{
\TNorm{(\matY\transp \hat\matS\matD )^{\dagger}}
=
\TNorm{(\matV_{\matY} \matSig_{\matY} \matU_{\matY}\transp \hat\matS \matD)^{\dagger} } 
{\buildrel (a)\over =}
\TNorm{(\matU_{\matY}\transp \hat\matS \matD)^{\dagger} (\matV_{\matY} \matSig_{\matY})^{\dagger}}
&\le&
\TNorm{(\matU_{\matY}\transp \hat\matS \matD)^{\dagger}} \TNorm{ (\matV_{\matY} \matSig_{\matY})^{\dagger}}\\
&\le&
(1-\sqrt{\rho_{\matY}}/{r})^{-1} \TNorm{ (\matV_{\matY} \matSig_{\matY})^{\dagger}}.
}
(a) uses
Lemma ~\ref{lem:pseudo}. To conclude, use $\TNorm{(\matY\transp \matS\matD )^{\dagger}}  = \TNorm{(\matD \matS \matY )^{\dagger}};$
\math{\TNorm{(\matV_\matY\matSig_\matY)^{\dagger}}=\TNorm{(\matY\transp)^{\dagger}}  = \TNorm{(\matY)^{\dagger}}}.

\end{proof}

\end{document}